\documentclass[letterpaper, 11pt]{article}
\usepackage[usenames]{xcolor}
\usepackage{xifthen}
\usepackage{algorithmic}
\usepackage{comment}
\usepackage{tcolorbox}
\usepackage{relsize}
\usepackage[font=footnotesize]{caption}
\usepackage[normalem]{ulem}
\usepackage{bm}
\usepackage{graphicx} 
\usepackage{float} 
\usepackage{subfigure} 
\usepackage{enumitem}
\usepackage{diagbox}

\usepackage{scrextend}
\usepackage{listings}
\usepackage[font=itshape]{quoting}

\def\fontsettingup{2} 

\usepackage{amsmath, amsthm, amssymb, amstext}
\usepackage[T1]{fontenc}
\ifthenelse{\fontsettingup = 1}{ \usepackage{eulerpx, eucal, tgpagella}   }{}
\ifthenelse{\fontsettingup = 2}{ \usepackage{mathpazo, eufrak, tgpagella} }{}

\usepackage{hyperref}
\hypersetup{
  colorlinks=true,
  linkcolor=blue,
  filecolor=magenta,
  urlcolor=cyan,
  citecolor=blue
}

\usepackage[linesnumbered,boxed,ruled,vlined, algosection]{algorithm2e}
\usepackage{tikz}
\usepackage{cleveref}
\usepackage{makecell}
\usepackage{multirow}

\usepackage[a4paper]{geometry}
\newgeometry{
textheight=9in,
textwidth=6.8in,
top=1in,
headheight=14pt,
headsep=25pt,
footskip=30pt
}

\newtheorem{theorem}{Theorem}

\newtheorem*{claim*}{Claim}

\newtheorem{fact}[theorem]{Fact}
\newtheorem{lemma}[theorem]{Lemma}
\newtheorem{proposition}[theorem]{Proposition}
\newtheorem{corollary}[theorem]{Corollary}
\theoremstyle{definition}

\newtheorem{definition}[theorem]{Definition}
\newtheorem{remark}[theorem]{Remark}
\newtheorem*{remark*}{Remark}

\ifthenelse{\fontsettingup = 1}{
  \def\*#1{\mathbf{#1}} 
  \def\+#1{\mathcal{#1}} 
  \def\-#1{\mathrm{#1}} 
  \def\^#1{\mathbb{#1}} 
  \def\!#1{\mathfrak{#1}} 
}{}

\ifthenelse{\fontsettingup = 2}{
  \def\*#1{\boldsymbol{#1}} 
  \def\+#1{\mathcal{#1}} 
  \def\-#1{\mathrm{#1}} 
  \def\^#1{\mathbb{#1}} 
  \def\!#1{\mathfrak{#1}} 
}{}

\def\oPr{\mathbf{Pr}}
\renewcommand{\Pr}[2][]{ \ifthenelse{\isempty{#1}}
  {\oPr\left[#2\right]}
  {\oPr_{#1}\left[#2\right]} } 

\def\oE{\mathbb{E}}
\newcommand{\E}[2][]{ \ifthenelse{\isempty{#1}}
  {\oE\left[#2\right]}
  {\oE_{#1}\left[#2\right]} }

\DeclareMathOperator*{\oVar}{\mathbf{Var}}
\newcommand{\Var}[2][]{ \ifthenelse{\isempty{#1}}
  {\oVar\left[#2\right]}
  {\oVar_{#1}\left[#2\right]} }

\def\oEnt{\mathbf{Ent}}
\newcommand{\Ent}[2][]{ \ifthenelse{\isempty{#1}}
  {\oEnt\left[#2\right]}
  {\oEnt_{#1}\left[#2\right]} }

\renewcommand{\epsilon}{\varepsilon}
\renewcommand{\emptyset}{\varnothing}

\let\epsilon=\varepsilon

\newcommand{\prob}[1]{\ensuremath{\text{{\bf Pr}$\left[#1\right]$}}}

\newcommand{\floor}[1]{\ensuremath{\left\lfloor#1\right\rfloor}}

\newcommand{\polylog}{\operatorname{polylog}}

\title{Differentially Private Algorithms for Graph Cuts: A Shifting Mechanism Approach and More}

\author{Rishi Chandra\thanks{Johns Hopkins University, US. \textnormal{E-mails: \url{rchand18@jhu.edu}, \url{mdinitz@cs.jhu.edu} }}
  \and
Michael Dinitz\thanks{Supported in part by NSF awards CCF-1909111 and CCF-2228995.} \footnotemark[1]
\and 
Chenglin Fan\thanks{Seoul National University. Part of the work was done at JHU. \textnormal{E-mail: \url{fanchenglin@gmail.com}}}
\and 
  Zongrui Zou\thanks{Nanjing University, China. \textnormal{E-mail: \url{zou.zongrui@smail.nju.edu.cn}}.}
}

\date{}

\usepackage[colorinlistoftodos,prependcaption,textsize=tiny]{todonotes}

\begin{document}

\begin{titlepage}
\maketitle
\thispagestyle{empty}
\begin{abstract}

In this paper, we address the challenge of differential privacy in the context of graph cuts, specifically focusing on the {\em multiway cut} and the {\em minimum $k$-cut}.  We introduce edge-differentially private algorithms that achieve nearly optimal performance for these problems.

Motivated by multiway cut, we propose the \textit{shifting mechanism}, a general framework for private combinatorial optimization problems.  This framework allows us to develop an efficient private algorithm with a multiplicative approximation ratio that matches the state-of-the-art non-private algorithm, improving over previous private algorithms that have provably worse multiplicative loss. We then provide a tight information-theoretic lower bound on the additive error, demonstrating that for constant $k$, our algorithm is optimal in terms of the privacy cost. The shifting mechanism also allows us to design private algorithm for the \textit{multicut} and \textit{max-cut} problems, with runtimes determined by the best non-private algorithms for these tasks. 

For the minimum $k$-cut problem we use a different approach, combining the exponential mechanism with bounds on the number of approximate $k$-cuts to get the first private algorithm with optimal additive error of $O(k\log n)$ (for a fixed privacy parameter). We also establish an information-theoretic lower bound that matches this additive error. Furthermore, we provide an efficient private algorithm even for non-constant $k$, including a polynomial-time 2-approximation with an additive error of $\tilde{O}(k^{1.5})$.

\end{abstract}

\end{titlepage}
\section{Introduction}\label{sec:intro}

 \vspace{-0.15cm}
 Graph analysis is one of the central topics in theoretical computer science. Given that a graph may represent real entities with privacy concerns, ensuring the privacy of individuals while extracting useful insights from graphs is of paramount concern. Differential privacy~\cite{dwork2006calibrating} is a rigorous mathematical framework designed to provide strong privacy guarantees for individuals in a dataset. An algorithm $\mathcal{A}$ is $(\epsilon,\delta)$-\emph{differentially private} if the probabilities of obtaining any set of possible outputs $S$ of $\mathcal{A}$ when run on two ``neighboring'' inputs $G$ and $G'$ are similar:
$\prob{\mathcal{A}(G) \in S} \leq e^{\epsilon} \cdot \prob{\mathcal{A}(G') 
\in S} + \delta$. When $\delta = 0$, we say that $\mathcal{A}$ preserves \emph{pure differential privacy}. While much work on differentially private algorithms has focused on numeric databases or learning settings, there has also been a significant amount of work on private graph analytics, e.g.,~\cite{gupta2012iterative, DBLP:conf/innovations/BlockiBDS13,DBLP:conf/tcc/KasiviswanathanNRS13,DBLP:conf/focs/BunNSV15,DBLP:conf/nips/AroraU19, DBLP:conf/nips/UllmanS19,DBLP:conf/focs/BorgsCSZ18,eliavs2020differentially,bun2021differentially, DBLP:conf/icml/NguyenV21,DBLP:journals/corr/abs-2106-00508,DBLP:conf/nips/Cohen-AddadFLMN22,DBLP:conf/nips/Fan0L22, DBLP:conf/soda/ChenG0MNNX23,
DBLP:conf/wads/DengGUW23, DBLP:conf/icml/ImolaEMCM23, DBLP:conf/nips/ChenCdEIST23,DBLP:journals/corr/abs-2308-10316,  DBLP:conf/nips/DalirrooyfardMN23,liu2024optimal}.  

Cuts are a fundamentally important object of study in graph theory and graph algorithms, as they are naturally relevant to clustering, community detection, and other important tasks.  In the private setting, there is an important line of work in which we attempt to essentially release the values of {all} cuts by privately answering \textit{cut queries} of the form ``given an $(S, T)$ cut, what is its weight?''~\cite{gupta2012iterative, blocki2012johnson, eliavs2020differentially, liu2024optimal}. However, there is a different perspective on cuts based on optimization: instead of trying to privately estimate the \emph{sizes} of {all} cuts, we instead wish to return the ``best'' cut, for some notion of best.  This includes, for example, classical problems such as max-cut, min $k$-cut, multiway cut, multicut etc.  

While there is some important work on private versions of these classical optimization problems, they have not been explored nearly as much as answering cut size queries privately.  A key difficulty for all of these problems is to output a \emph{structure}, rather than just a value.  It is easy to privately compute the \emph{cost} of the optimal multiway cut or $k$-cut using, e.g., the Laplace mechanism.  But for many tasks (community detection, clustering, etc.) we want the cut itself.  We cannot output the edges in the cut as they are part of the input dataset, but for many applications what we want is actually the vertex partition, and returning such a partition does not (at least a priori) violate privacy.  So that is the structure that our algorithms will return. 

In this paper we explore private versions of several classical cut problems: multiway cut, multicut, max-cut and $k$-cut.  In multiway cut we are given a collection of terminals $T$ and are asked to find the cheapest cut separating all of the terminals, in multicut we are given a collection of $(s_i, t_i)$ pairs and are asked to find the cheapest cut separating each pair, in max-cut we are asked to find the globally maximum weight cut, and in $k$-cut we are asked to find the minimum weight cut that separates the graph into at least $k$ connected components.  See Section~\ref{sec:prelims} for detailed definitions.

We particularly focus on multiway cut and $k$-cut, where we show that our algorithms provide essentially optimal bounds.  All of these problems have been studied extensively in the non-private setting (see Appendix~\ref{sec:related} for more discussion).  Very recently, private algorithms have been developed for private multiway cut, most notably by~\cite{DBLP:conf/nips/DalirrooyfardMN23}, who gave an algorithm with optimal additive loss but with a multiplicative loss of $2$.  This is in contrast to the non-private setting, where the best-known multiplicative loss is $1.2965$~\cite{sharma2014multiway}.  This extra multiplicative loss can easily overwhelm the additive loss when $\mathsf{OPT}$ is large (e.g., in the presence of edge weights). Minimum $k$-cut, on the other hand, has not been studied in the context of privacy.  This is somewhat surprising, since global min-cut (or minimum $2$-cut) was one of the first combinatorial optimization problems to be studied in the context of privacy~\cite{gupta2010differentially}.  So we initiate the study of private $k$-cut in this work. 

From a technical point of view, for general private combinatorial problems, the most common approaches in the literature are \textit{input perturbation} (e.g., private min $s$-$t$ cut \cite{DBLP:conf/nips/DalirrooyfardMN23}), or using the exponential mechanism (e.g., private min-cut~\cite{gupta2010differentially} or private topology selection~\cite{liu2024optimal}), since the outputs are typically discrete structures rather than real-valued vectors.  In this paper we study both of these approaches: new versions of input perturbation for multiway cut, max-cut, and multicut, and new analyses of the exponential mechanism (and an efficient variation) for $k$-cut.

\paragraph{Input Perturbation.}
The most common form of input perturbation is usually generating a private synthetic dataset or graph by adding either (additive or multiplicative) independent noise~\cite{blocki2012johnson, dwork2014analyze, upadhyay2021differentially} or correlated noise, such as private multiplicative weights update~\cite{gupta2012iterative, eliavs2020differentially}. Then, a non-private cut algorithm is applied to the private graph as a post-processing step. This way of adding noise can be shown to be optimal for some problems, such as giving optimal $\tilde{\Theta}(\sqrt{mn})$ additive error for approximating the sizes of all cuts on a weighted $n$-vertex, $m$-edge graph~\cite{liu2024optimal}. However, for many other problems, it is excessive and it loses utility in the sense that we only want some combinatorial structures instead of a whole synthetic graph. For instance, \cite{DBLP:conf/nips/DalirrooyfardMN23} presents an $O(n\log n)$ additive error algorithm for approximating the min $s$-$t$ cut, achieving significantly lower error compared to the error incurred when releasing the sizes of all cuts by a synthetic graph. Therefore, a central question of this paper is to determine when it is necessary to perturb everything in the input to generate a fully private synthetic graph, and when adding less noise is sufficient. In particular, we ask the following question:


\vspace{-0.1cm}
\begin{quote}
  {\em Main Question: Given a combinatorial problem $P$, what property of $P$ determines \textbf{how much} and \textbf{where} the noise needs to be added in order to privately solve it with minimal disturbance?}
\end{quote}
\vspace{-0.1cm}

To answer this question, we present the definition of a \textit{dominating set} (Definition \ref{def:dominating_set}) of a well-defined combinatorial optimization problem. Building on this definition, we propose the \textit{shifting mechanism} as a general framework for input perturbation (Section \ref{sec:shifting}). 

Notably, this mechanism not only encompasses the previous algorithm for the minimum $s$-$t$ cut problem~\cite{DBLP:conf/nips/DalirrooyfardMN23} as a specific instance, but it also directly leads to the design of private algorithms for a wide range of other cut problems.  In particular, it leads to our new algorithms for multiway cut, multicut, and max cut.  Furthermore, in all of these applications, the error bounds provided by the shifting mechanism are either proven to be optimal or matching with those of the exponential mechanism in terms of the dependency on $|V|$, where $V$ is the vertex set. 



\paragraph{Exponential Mechanism.}
The other common approach to private combinatorial optimization is the exponential mechanism.  This mechanism is a particularly natural fit for combinatorial optimization, since it naturally operates on discrete objects.  For minimum $k$-cut, one of the most fundamental cut and clustering problems in graph algorithms, we show that the exponential mechanism leads to extremely good bounds.  Unfortunately the exponential mechanism is not efficient, so we also develop efficient versions of our algorithm that incur some extra loss.  


 \subsection{Our results} 
We consider the standard notion of \emph{edge-level} privacy~\cite{DBLP:conf/icdm/HayLMJ09}, in which $G$ and $G'$ are neighboring graphs if they differ by one edge weight of $1$ (see Section~\ref{sec:prelims} for a formal definition)\footnote{For unweighted graphs, edge-level privacy corresponds to adding or deleting a single edge. Consequently, the graph's topology remains unknown to the public, which contrasts with weight-level privacy~\cite{DBLP:conf/pods/Sealfon16}.}. Our main results include edge-level differentially private algorithms for various cut problems using the shifting mechanism, most notably for multiway cut. Additionally, we also provide tight upper and lower bounds for private minimum $k$-cut (abbreviated as $k$-cut). 

\subsubsection{Application of the shifting mechanism in private multiway cut}

Private multiway cut was first studied in~\cite{DBLP:conf/nips/DalirrooyfardMN23}, where they presented a differentially private algorithm for min $s$-$t$ cut (i.e., $k = 2$) with an optimal additive error of $\tilde{O}(n/\epsilon)$. For $k\geq 3$, the non-private multiway cut problem becomes NP-hard, and they show that the private version of this problem can be solved by running at most $O(\log k)$ rounds of the private algorithm for min $s$-$t$ cut.  This yields a private approximation for multiway cut with a multiplicative approximation of $2$ and an additive loss of $O(n\log^2 k/\epsilon)$, i.e., the algorithm returns a solution of cost at most $2\cdot \text{OPT} + O(n\log^2 k/\epsilon)$ (see \cite[Theorem 4.2]{DBLP:conf/nips/DalirrooyfardMN23}). 

However, the 2-approximation ratio is far from optimal. The best non-private approximation ratio for multiway cut is achieved by solving the simplex embedding linear program of~\cite{DBLP:journals/jcss/CalinescuKR00} and then rounding the fractional solution, which gives a $(\approx 1.3)$ approximation ratio~\cite{sharma2014multiway}. While $2$ and $1.3$ might not seem far apart when the additive error is $\Omega(n)$, in the presence of edge weights the multiplicative loss can easily outweigh the additive loss.  Thus it is important to minimize the multiplicative loss in addition to minimizing the additive loss. 

Unfortunately, the state-of-the-art rounding schemes used by~\cite{sharma2014multiway} are quite complex, so to achieve their approximation ratio privately a natural idea is to solve the linear program privately, since then any rounding scheme functions as post-processing and so we do not need to worry about making it private.  However, the obvious attempts under this idea do not yield the optimal $O(n/\epsilon)$ error, even if $k\geq 3$ is a constant. For example, a folklore approach is to add Laplace or Gaussian noise between each pair of vertices, then find a fractional multiway cut on the noisy graph. For fixed $k$, this gives $\tilde{O}(n^{1.5}/\epsilon)$ additive error because of the concentration of noise and the union bound among all cuts. Moreover, adding noise might result in negative edge weights, making the simplex embedding linear program not polynomial time solvable. Alternatively, Hsu et al.~\cite{hsu2014privately} provide a general method for solving linear programs with approximate differential privacy. For fixed $k$, their approach also results in an $\tilde{O}(n^{1.5}/\epsilon)$ error for the simplex embedding linear program of multiway cut\footnote{Solving the simplex embedding linear program of multiway cut under edge-level DP corresponds to the objective private LP in \cite{hsu2014privately}. In the simplex embedding program, there are $nk$ variables and the sum of them is $n$. Applying the upper bound for the objective private LP in Section 4.5 of \cite{hsu2014privately} gives an $\tilde{O}(n\sqrt{nk})$ additive error.}. Therefore, we have the following natural question:
\begin{quote}
  {\em Question 2: For fixed $k\geq 3$, is it possible to privately and efficiently solve the multiway cut problem with both \textbf{optimal additive error} and the \textbf{best-known non-private multiplicative approximation ratio}?}
\end{quote}
In this paper, as the most notable application of our shifting mechanism, we provide a positive answer to this question for any constant $k\geq 3$ by showing that the particular structure of the~\cite{DBLP:journals/jcss/CalinescuKR00} relaxation allows us to solve it with \textit{pure} differential privacy and smaller additive loss. In particular, we give the following theorem:

\begin{theorem}[Informal]\label{thm:result_1}
For any weighted graph $G$, there exists a polynomial time $(\epsilon,0)$-differentially private algorithm for multiway cut which outputs a solution with cost at most $1.2965 \cdot \mathsf{OPT}+\tilde{O}(nk/\epsilon)$, 
  where $\mathsf{OPT}$ is the value of the optimal multiway cut.
\end{theorem}

Note that the approximation ratio for the private multiway cut in Theorem \ref{thm:result_1} matches that of the state-of-the-art non-private algorithms. Additionally, for any fixed $k\geq 3$, we show that the additive error of our algorithm is optimal with respect to the dependence on $n$ and $\epsilon$ by proving the following lower bound for pure differential privacy, which applies even to non-efficient algorithms. 

\begin{theorem}[Informal]\label{thm:result_lower_bound}
  Fix any $2\leq k \leq n/2$. Any $(\epsilon,0)$-differentially private algorithm for multiway cut on $n$-vertex graphs has expected additive error at least $\frac{n\log(k/3)}{24\epsilon}$.
\end{theorem}
The above information-theoretic lower bound on the additive error for private multiway cut is optimal in terms of $n$, $k$, and $\epsilon$, as it asymptotically matches the upper bound given by an inefficient exponential mechanism. It also generalize the lower bound in \cite{DBLP:conf/nips/DalirrooyfardMN23} from $k = 2$ to all $k =O(n)$. We remark that for $k = O(1)$, this lower bound implies that no private algorithm achieves non-trivial error on unweighted graphs. Including the results of this paper, we summarize the current known conclusions about private multiway cut problem in Table \ref{tab:multiway-cut}.
\begin{table*}[t]
    \centering
    \caption{Current known results for private multiway cut with fixed $k$.}\label{tab:multiway-cut}
    \begin{tabular}{|c|c|c|c|c|c|}
        
        \hline
        Value of $k$ &  \makecell[c]{Category} &  \makecell[c]{Reference} & \makecell[c]{Privacy} & Error rate & Efficient?\\
        \hline 

        \multirow{2}*{$k = 2$ }& Upper bound & \cite{DBLP:conf/nips/DalirrooyfardMN23}  &  $(\epsilon,0)$-DP & $O(n/\epsilon)$ & Yes \\
        \cline{2-6}
        ~& Lower bound & \cite{DBLP:conf/nips/DalirrooyfardMN23} &$(\epsilon,0)$-DP & $\Omega(n/\epsilon)$ & --- \\
        \hline
       \multirow{5}*{$k \geq 3$ }& Upper bound & \makecell[c]{Exponential\\ mechanism}  &  $(\epsilon,0)$-DP & $ {O}_k(n/\epsilon)$ & No \\
        \cline{2-6}
        ~ & Upper bound & \cite{hsu2014privately}  &  $(\epsilon,\delta)$-DP & $(\approx 1.3)\mathsf{OPT} + \tilde{O}_k(n^{1.5}/\epsilon)$ & Yes \\
        \cline{2-6}
         ~ & Upper bound & \cite{DBLP:conf/nips/DalirrooyfardMN23}  &  $(\epsilon,0)$-DP & $2\mathsf{OPT} + {O}_k(n/\epsilon)$ & Yes \\
        \cline{2-6}
         ~ & Upper bound &  Theorem \ref{thm:result_1} &  $(\epsilon,0)$-DP & $(\approx 1.3)\mathsf{OPT} + {O}_k(n/\epsilon)$ & Yes \\
        \cline{2-6}
        ~& Lower bound & Theorem \ref{thm:result_lower_bound} &$(\epsilon,0)$-DP & $\Omega_k(n/\epsilon)$ & --- \\
  \hline
    \end{tabular}
  \end{table*}
\subsubsection{Other applications of the shifting mechanism}\label{sec:results_shifting}
In addition to multiway cut, we also show in Section \ref{sec:app_shifting} that the shifting mechanism can be used for other cut problems, and offers a unified framework that encompasses previous studies.
\begin{enumerate}
    \item \textbf{Private min $s$-$t$ cut.} Let $x\in \mathbb{R}_{+}^{{n\choose 2}}$ be a non-negative vector that encodes a weighted graph. We define $f(x)$ to be the unique\footnote{It is easy to show that by adding negligible perturbation (with sufficient precision) to the edge weights of a graph, the minimum $s$-$t$ cut will be unique.  See also the discussion in \cite{DBLP:conf/nips/DalirrooyfardMN23}.} min $s$-$t$ cut in the graph represented by $x$. We show that the collection of edges crossing between terminals $\{s,t\}$ and any non-terminals builds up a dominating set (Definition \ref{def:dominating_set}) of $f$ with sensitivity $2$. This directly implies a private algorithm for finding min $s$-$t$ cut with a \textit{purely} additive error $\tilde{O}(n)$, matching the lower bound giving in \cite{DBLP:conf/nips/DalirrooyfardMN23}. Thus, our shifting mechanism encompasses the previous work on min $s$-$t$ cut~\cite{DBLP:conf/nips/DalirrooyfardMN23} as a special instance, and also offers a significantly different and more streamlined analysis.
    \item \textbf{Private multicut.} Given a weighted and undirected graph $G$ together with a collection of pairs $T = \{(s_1,t_1), \cdots (s_k,t_k)\}$ known as terminals, the multicut problem is to find a minimum cut (or equivalently, a partition of the vertices) that separates each pair in $T$. As one of the natural generalizations of the min $s$-$t$ cut problem, multicut plays a central role in the approximation algorithms literature. Yet to the best of our knowledge, the private version of multicut for $k\geq 2$ has not been studied before. By applying the shifting mechanism, we directly give an additive noise, polynomial-space algorithm with $\tilde{O}(k(n+k))$ additive error, matching the error of the non-efficient exponential mechanism on this problem for any constant $k$. Further, our algorithm runs in polynomial time when $k = 1,2$.  So when $k=2$, this is the first efficient private algorithm that has the same error as the exponential mechanism.
    \item \textbf{Private max-cut.} Finally, we show that our shifting mechanism also gives an algorithm for finding the maximum cut of a weighted graph. The error of our algorithm is $\tilde{O}(n)$, matching that of the exponential mechanism on private max-cut. Unlike the exponential mechanism, the time and space consumption of our algorithm depend on the complexity of finding the exact max-cut on a noisy graph with non-negative edge weights. Although this does not yet provide an efficient algorithm, it enables the use of various faster algorithms for finding max-cut in polynomial space~\cite{scott2007linear, golovnev2011new}, thereby accelerating the exponential mechanism within polynomial space.
\end{enumerate}

\subsubsection{Private minimum $k$-cut}
The private minimum cut problem (or minimum $2$-cut) was first introduced by~\cite{gupta2010differentially}, who proposed a $(\epsilon,0)$-private algorithm that incurs an expected additive error of $\Theta(\log n/\epsilon)$. This additive error is far below the $\Omega(n/\epsilon)$ lower bound for private min $s$-$t$ cut proved by~\cite{DBLP:conf/nips/DalirrooyfardMN23}, showing that the two problems are fundamentally different and require different techniques.  
 Their algorithm consists of two stages:
 (1) adding random edges to the input graph $G$ in a private manner to increase the 
 optimal cut cost to 
 $\Omega(\log n/\epsilon)$,
 and (2) applying the exponential mechanism over all cuts in the augmented graph.  There is also a matching lower bound, and so the private min-cut problem is essentially solved.  On the other hand, there are no known bounds (upper or lower) for minimum $k$-cut.  Therefore, we ask the following question:

 \begin{quote}
  {\em Question 3: What is the ``correct'' dependence on $k$, $n$, and $\epsilon$ that characterizes the achievable error of the private minimum $k$-cut problem?}
\end{quote}

To answer the above question, we first give an upper bound on the private minimum $k$-cut problem. In particular, we give algorithms where the additive loss is a function of the number of approximate $k$-cuts. 
 Bounds on this number have been studied extensively~\cite{DBLP:journals/jacm/Karger00,chekuri18, DBLP:journals/jacm/GuptaHLL22}, and when we plug in the known bounds we get the following theorem:

 \begin{theorem}[Informal]\label{thm:result_kcut_upper}
    For any unweighted graph $G$, there is an inefficient $(\epsilon,0)$-DP algorithm for $k$-cut which outputs a solution with cost at most $\mathsf{OPT} + O(k\log{n} / \epsilon)$ with high probability. Further, allowing approximate differential privacy, there is an  $(\epsilon,\delta)$-DP algorithm achieving the same utility in time $O(n^{O(k)}(\log n)^{O(k^2)})$. Here, $\mathsf{OPT}$ is the value
of the minimum  $k$-cut on $G$.
\end{theorem}
Our approach is also optimal in terms of the error it achieves, even up to logarithmic factors. Specifically, we complement our upper bound with the following tight lower-bound result:
\begin{theorem}[Informal]\label{thm:result_lower_bound_k_cut}
   Fix any $2\leq k \lesssim n^{1/2}$. Any $(\epsilon,0)$-differentially private algorithm for approximating  $k$-cut on $n$-vertex graphs has expected error  $\Omega(k \log n/\epsilon)$.
\end{theorem}

Combining Theorem \ref{thm:result_kcut_upper} and Theorem \ref{thm:result_lower_bound_k_cut} answers Question 3. Unfortunately, our pure DP algorithm is inefficient, and our approximate DP algorithm is inefficient when $k = \omega(1)$ (which is unsurprising since $k$-cut is NP-hard when $k = \omega(1)$).  So we complement Theorem~\ref{thm:result_kcut_upper} with an efficient algorithm with sub-optimal bounds. 
\begin{theorem}[Informal]\label{thm:result_efficient}
For any weighted graph $G$,
  there exists a polynomial time $(\epsilon,\delta)$-differentially private algorithm for  $k$-cut which outputs a solution with cost at most $2 \cdot \mathsf{OPT}+\tilde{O}(k^{1.5}/\epsilon)$, 
  where   $\mathsf{OPT}$ is the value of the minimum  $k$-cut.
\end{theorem}

We remark that unlike multiway cut, the $2$-approximation for the minimum $k$-cut is essentially optimal~\cite{manurangsi2017inapproximability}, assuming the Small Set Expansion Hypothesis (SSEH). Including the results of this paper, we summarize the current known bounds for the private minimum $k$-cut problem in Table \ref{tab:k-cut}.

\begin{table*}[h]
    \centering
    \caption{Current known results for the private minimum  $k$-cut problem.}\label{tab:k-cut}
    \begin{tabular}{|c|c|c|c|c|c|}
        
        \hline
        Value of $k$ &  \makecell[c]{Category} &  \makecell[c]{Reference} & \makecell[c]{Privacy} & Error rate & Efficient?\\
        \hline 

        \multirow{2}*{$k = 2$ }& Upper bound & \cite{gupta2010differentially}  &  $(\epsilon,\delta)$-DP & $O(\log n/\epsilon)$ & Yes \\
        \cline{2-6}
        ~& Lower bound & \cite{gupta2010differentially} &$(\epsilon,0)$-DP & $\Omega(\log n/\epsilon)$ & --- \\
        \hline
       \multirow{3}*{$k \geq 3$ }&  Upper bound & Theorem \ref{thm:result_kcut_upper}  &  $(\epsilon,\delta)$-DP & $O(k\log n/\epsilon)$ &  \makecell[c]{Yes\\(for constant $k$)} \\
        \cline{2-6}
        ~ &    Lower bound & Theorem \ref{thm:result_lower_bound_k_cut} &$(\epsilon,0)$-DP & $\Omega(k\log n/\epsilon)$ & --- \\
        \cline{2-6}
         ~ & Upper bound & Theorem \ref{thm:result_efficient}  &  $(\epsilon,\delta)$-DP & $2\mathsf{OPT} + \tilde{O}(k^{1.5}/\epsilon)$ & Yes \\

        \cline{2-6}
  \hline
    \end{tabular}
  \end{table*}


\section{Preliminaries} \label{sec:prelims}

Here, we give the formal definition of edge-level differential privacy. 

\begin{definition}[Edge-level differential privacy]\label{def:dp}
    Let $\mathcal{A}:\mathcal{D}\rightarrow \mathcal{R}$ be a randomized algorithm, where $\mathcal{R}$ is the output domain. For fixed $\epsilon >0$ and $\delta\in [0,1)$, we say that $\mathcal{A}$ preserves $(\epsilon,\delta)$-differential privacy if for any measurable set $S\subseteq \mathcal{R}$ and any pair of neighboring graphs $G,G'\in \mathcal{D}$, 
    $$\mathbf{Pr}[\mathcal{A}(G)\in S] \leq \mathbf{Pr}[\mathcal{A}(G')\in S]\cdot e^{\epsilon} + \delta.$$
    If $\delta = 0$, we also say $\mathcal{A}$ preserves pure differential privacy (denoted by $\epsilon$-DP). Here, $G=(V,E,w)$ and $G'=(V,E',w')$ with $w,w' \in \mathbb R_{+}^{N}$ are {\em neighboring} if $\| w - w'\|_0 \leq 1$ and $\| w- w'\|_\infty \leq 1$. That is, $G$ and $G'$ differ in one edge by weight at most {\em one}. 
\end{definition}
Important properties of differential privacy that we use include post-processing, adaptive and advanced composition, and the Laplace and exponential mechanisms.  Since these are standard, we defer their discussion to Appendix~\ref{app:DP}. 

Given a graph $G = (V, E, w)$ and a partition $\mathcal S = (S_1, S_2, \dots, S_k)$ of $V$, let $\delta(\mathcal S) = \{e \in E : \text{endpoints of $e$ are in different parts of } \mathcal S\}$. The problems we study include Multiway Cut, Minimum k-Cut, Multicut and Max-Cut.  
They are formally defined as follows (we note that as for most cut problems there are equivalent definitions in terms of edges and node partitions; due to edge-level privacy, we will use the node partition-based definitions).  


\begin{definition}
    In the \emph{Multiway Cut} problem, we are given a weighted graph $G = (V, E, w)$ and a collection of terminals $T = \{t_1, t_2, \dots, t_k\} \subseteq V$.  A feasible solution is a partition $\mathcal S$ of $V$ into $k$ parts $(S_1, S_2, \dots, S_k)$ such that $t_i \in S_i$ for all $i \in [k]$.   The cost of a partition is  $\sum_{e \in \delta(\mathcal S)} w(e)$, and we want to find the feasible solution of minimum cost.
\end{definition}

\begin{definition}
    In the \emph{Minimum $k$-Cut} problem (or just \emph{$k$-cut}), we are given a weighted graph $G = (V, E, w)$.  A feasible solution is a partition $\mathcal S$ of $V$ into $k$ parts $(S_1, S_2, \dots, S_k)$.  The cost of a partition is  $\sum_{e \in \delta(\mathcal S)} w(e)$, and we want to find the feasible solution of minimum cost.
\end{definition}

\begin{definition}
    In the \emph{Multicut} problem, we are given a weighted graph \( G = (V, E, w) \) and a collection of pairs of vertices \( (s_1, t_1), (s_2, t_2), \dots, (s_k, t_k) \) where \( s_i, t_i \in V \) for all \( i \in [k] \). 
    A feasible solution is a partition \( \mathcal{S} = (S_1, S_2, \ldots, S_l) \) of the vertex set \( V \) such that for each pair \( (s_i, t_i) \) for all \( i \in [k] \), it holds that \( s_i \in S_j \) and \( t_i \in S_{j'} \) for some \( j, j' \in [l] \) with \( j \neq j' \). The cost of a partition is \( \sum_{e \in \delta(\mathcal{S})} w(e) \), and we seek to find a feasible solution of minimum cost.
\end{definition}

\begin{definition}
    In the \emph{Max-Cut} problem, we are given a weighted graph \( G = (V, E, w) \). A feasible solution is a partition \( \mathcal{S} \) of \( V \) into two disjoint subsets \( S_1 \) and \( S_2 \). The objective is to maximize the total weight of the edges crossing the cut, defined as \( \sum_{e \in \delta(\mathcal{S})} w(e) \), where \( \delta(\mathcal{S}) = \{ e \in E : e \text{ connects a vertex in } S_1\\
    \text{ to a vertex in } S_2 \} \).
\end{definition}

\section{The Shifting Mechanism and Its Analysis}\label{sec:shifting}
In this section, we introduce our general framework for differentially private combinatorial optimization, and provide its privacy analysis.
\subsection{Problem formalization}
Let $x\in \mathcal{X}^N$ be the dataset, where $\mathcal{X}\subseteq \mathbb{R}$ is the data domain. The goal is to privately output $f(x)$ for a function $f: \mathcal{X}^N \rightarrow \mathcal{R}$. Two datasets are neighboring if and only if they have a difference on one coordinate by one. In all of our applications, $\mathcal{X} = \mathbb{R}_{\geq 0}$ and it represents the edge weight in a $n$-vertex graph where $N = {n\choose 2}$. In this case, the privacy notion then corresponds to edge-level privacy. Notably $f(x)$ is usually a structure rather than a value. For instance, consider the case where $\mathcal{X}^N$ denotes a collection of graphs with a unique min-cut, and $\mathcal{R}$ represents the set of all $(S,V\backslash S)$ cuts in an $n$-vertex graph. Then, $f(x)$ can be defined as the unique min-cut of the graph encoded by $x$. 

For any $x\in \mathcal{X}^N$ and $S\subseteq [N]$, we use $x|_S$ to denote the vector in $\mathcal{X}^S$, whose values are taken from the coordinates of $x$ specified by $S$. Given any $S$, we may also write $f(x)$ as $f(x|_S, x|_{[N]\backslash S})$. Consider any function $f: \mathcal{X}^N \rightarrow \mathcal{R}$. The following definition characterizes the ``hardness'' of approximating $f(x)$ in terms of our shifting mechanism:
\begin{definition}\label{def:dominating_set}
Fix any $S\subseteq [N]$ and $s\geq 0$, a function $f:\mathcal{X}^N \rightarrow \mathcal{R}$ admits a dominating set $S$ of sensitivity $s$ if for any pair of neighboring datasets $x,x'\in \mathcal{X}^N$, there exists a vector $a\in \mathcal{X}^S$ with $\|a\|_1 \leq s$ such that
\begin{itemize}
    \item $a$ is a function of $f(x)$ and $x'-x$.
    \item $f(x'|_S + a, x'|_{[N]\backslash S}) = f(x).$
\end{itemize}
\end{definition}
\noindent Note that we ask the vector $a\in \mathcal{X}^S$ must \textit{only} depend on the value of $f(x)$ and $x-x'$, instead of $x$ or $x'$ themselves. Intuitively, saying that a function admits a dominating set $S$ means that if we replace the input $x$ by a neighboring input $x'$, we can always shift $x'$ only in $S$ (by adding a ``correction'' vector $a$) based on (1) the value of $f(x)$ and (2) how $x$ is altered, ensuring that the output remains unchanged at $f(x)$. Further, $S$ is oblivious to the input.

Definition \ref{def:dominating_set} provides a new paradigm for input perturbation in differential privacy. That is, only perturbing values in the dominating set and computing $f$ on such noisy input, instead of perturbing every coordinate of $x$ and computing $f$ as post processing. In this sense, the difficulty of approximating $f(x)$ is captured by the quality of the dominating set that can be identified for $f$. 

To further elaborate this, we give some examples. According to the privacy notion, just letting $a = x-x'$ yields the fact that every function $f$ has a trivial dominating set $S = [N]$ of sensitivity $1$. In this case, adding noise in the dominating set is equivalent to adding independent noises on each coordinate, which preserves privacy for any $f$ by post-processing. Another extreme example is when $f$ is the averaging function where $f(x) = \frac{1}{N}\sum_{i=1}^N x^{(i)}$. Then, it is easy to verify that each single coordinate $\{i\}\subseteq [N]$ is a dominating set of sensitivity $1$. So we may as well let $S$ be the first coordinate. Clearly adding noise only on $S$ does not yields a private dataset, as a neighboring dataset could have a difference anywhere else. However, the final output is still $f(x) + \frac{1}{N}\text{Lap}(1/\epsilon)$, which is private according to the basic Laplace mechanism. 

In section \ref{sec:app_stcut}, \ref{sec:app_multicut}, and \ref{sec:app_maxcut}, we will provide more examples of function $f$ within the context of combinatorial optimization (e.g. outputting the unique max-cut of the input graph), in which finding the dominating set becomes non-trivial.

\subsection{The shifting mechanism}

In this section, we present our simple mechanism that utilizes the dominating set of a function $f$ and analyze its privacy. Suppose the function $f$ admits a dominating set $S$ of sensitivity $s$.

\begin{center}
    \begin{tcolorbox}[=sharpish corners, colback=white, width=1\linewidth]\label{alg:shifting}

        \begin{center}
      \textbf{\emph{Algorithm \ref*{alg:shifting}: The Shifting Mechanism}}
        \end{center}
        \vspace{6pt}
    \begin{enumerate}
  \item \textbf{Input:} A dataset $x\in \mathcal{X}^N$, privacy budget $\epsilon > 0$.
  \item Sample i.i.d Laplace noise $\{X_i\}_{i\in [S]}$, where $\forall i, X_i \sim \text{Lap}(s/\epsilon)$.
  \item Construct $\tilde{x}$ where for each $i\in S$, $\tilde{x}^{(i)} = x^{(i)} + X_i$, otherwise $\tilde{x}^{(i)} = x^{(i)}$.
  \item \textbf{Output:} $f(\tilde{x})$.
\end{enumerate}
  \end{tcolorbox} 
  \end{center}

\noindent We will show that Algorithm \ref{alg:shifting} preserves $\epsilon$-differential privacy, and thus as a meta-algorithm, it gives a new input perturbation framework for private problems: 

\begin{center}
    \begin{tcolorbox}[=sharpish corners, colback=white, width=1\linewidth]
        \begin{enumerate}
    \item Formalize the problem into a function $f$;
    \item Find a dominating set $S$ of $f$;
    \item Add noise in $S$ to generate a (non-private) synthetic dataset;
    \item Run $f$ on the (non-private) synthetic dataset, and output the (private) solution.
        \end{enumerate}
  \end{tcolorbox} 
  \end{center}

In Section \ref{sec:app_shifting}, we introduce several simple applications of the framework we proposed above, including the three examples in the context of private cut problems mentioned in Section \ref{sec:results_shifting}. Next, we present the privacy guarantee of Algorithm \ref{alg:shifting}. The privacy analysis of this algorithm is also as simple as it appears to be.

\begin{theorem}\label{thm:privacy_shifting}
If the function $f$ admits a dominating set $S$ of sensitivity $s$, then Algorithm \ref{alg:shifting} preserves $(\epsilon,0)$-differential privacy.
\end{theorem}

\begin{proof}
We note that the randomness of $f(\tilde{x})$ entirely comes from the randomness of the noise. That is, given any dataset $x\in \mathcal{X}^N$, the output of Algorithm \ref{alg:shifting} is determined by the value of the noises. Fix $x$ and $x'$ be any pair of neighboring datasets such that $\|x-x'\|_1\leq 1$ and $\|x-x'\|_0\leq 1$. Let the output $f(\tilde{x}) =\mathcal{A}(x)$, it suffices to show that
\begin{align}\label{eq:dp}
    \forall y\in \mathcal{R}, \Pr{\mathcal{A}(x) = y} \leq e^{\epsilon}\Pr{\mathcal{A}(x') = y}.
\end{align}
For any possible $y\in \mathcal{R}$, we define $\mathcal{S}_y\subseteq \mathbb{R}^{|S|}$ be a measurable set where\footnote{For the sake of simplicity, we consider the case where $\mathcal{R}$ is discrete. In the case where $\mathcal{R}$ is not, one could discretize it into sufficiently small intervals to ensure a small loss of accuracy.} 
$$\mathcal{S}_y = \{Z\in \mathbb{R}^{|S|}: f(x|_{S} + Z, x|_{[N]\backslash S}) = y \}.$$
Intuitively, $\mathcal{S}_y$ is the collection of noises that makes the output be $f(\tilde{x}) = y$. Similarly, we define $\mathcal{S}'_y = \{Z\in \mathbb{R}^{|S|}: f(x'|_{S} + Z, x'|_{[N]\backslash S}) = y \}$. We note that both $\mathcal{S}_y$ and $\mathcal{S}_{y'}$ can be empty. For any value of noises $Z\in \mathcal{S}_y$, let 
$$\hat{x} = (x|_S + Z, x|_{[N]\backslash S})$$
and $\hat{x}' = (x'|_S + Z, x'|_{[N]\backslash S})$. According to the fact that $S$ is a dominating set of $f$ with sensitivity $s$, by Definition \ref{def:dominating_set}, for all such $\hat{x}$ and $\hat{x}'$, there exists correction coefficients $a\in \mathbb{R}^S$ with $\|a\|_1\leq s$ such that
\begin{itemize}
    \item $a$ is decided by $f(\hat{x}) = y$ and $\hat{x} - \hat{x}' = x-x'$, and
    \item $f(x'|_S + Z + a, x'|_{[N]\backslash S}) = f(x|_{S} + Z, x|_{[N]\backslash S}) = y.$
\end{itemize}
Then, by the definition of $\mathcal{S}_y',$ it is equivalent to saying that for any $Z\in \mathcal{S}_y$, we have that $Z+a \in \mathcal{S}_y'$, and $a$ does not depend on the value of $Z$. Therefore, we have $\mathcal{S}_y + a \subseteq \mathcal{S}'_y$, and thus
\begin{align}\label{eq:step1}
    \Pr{\mathcal{A}(x') = y} = \Pr{Z\in \mathcal{S}_{y}'} \geq \Pr{Z\in \mathcal{S}_y+a}.
\end{align}
On the other hand, by the basic Laplace mechanism (Lemma \ref{lem:laplace}) and that $\|a\|_1 \leq s$, we also have
\begin{align}\label{eq:step2}
    \Pr{\mathcal{A}(x) = y} = \Pr{Z\in \mathcal{S}_{y}} \leq e^\epsilon\Pr{Z\in \mathcal{S}_y+a}.
\end{align}
We note that $\mathcal{S}_y$ and $\mathcal{S}_y'$ does not depend on the randomness of the noise, so the Laplace mechanism can be applied. Combining \Cref{eq:step1} and \Cref{eq:step2} completes the proof.
\end{proof}

\section{Private Simplex Embedding and Multiway Cut} \label{sec:multiway}

Our primary application of the shifting mechanism is to Multiway Cut. 
This is perhaps surprising, since the shifting mechanism as described in Section~\ref{sec:shifting} requires \emph{exact} algorithms for the underlying problem, which obviously do not exist for NP-hard problems such as Multiway Cut.  But by combining the shifting mechanism with nonprivate \emph{approximation} algorithms of Multiway Cut, and doing some extra analysis on how the particular algorithms fit together, we are able to get both privacy and performance comparable to the non-private setting.  

Specifically, we start from the observation that the state of the art approximation algorithms for Multiway Cut rely on rounding a particular LP relaxation due to~\cite{DBLP:journals/jcss/CalinescuKR00}.  So instead of applying the shifting mechanism to Multiway Cut, we instead apply it to find a private solution to this linear program.  Then we simply use the best-known rounding~\cite{sharma2014multiway} on our private solution, and the post-processing guarantees of differential privacy imply privacy with no extra loss. 

So in this section, we apply this framework to get the following theorem:

\begin{theorem}\label{thm:main_multiway_cut}
  There exists a polynomial time $(\epsilon,0)$-differentially private algorithm such that on any input weighted graph $G$, it outputs a solution $x\in \{1,2,\cdots, k\}^n$ on multiway cut of $G$ such that $$\mathcal{E}(x) = (\approx 1.3)\mathsf{OPT}^T(G) + \tilde{O}(n{k}/\epsilon),$$
  where $\mathsf{OPT}^T(G)$ is the optimal value of  multiway cut on $G$ with respect to terminals $T$, and $\mathcal{E}(x)$ is the sum of edge weights across the vertex partition defined by $x$. 
\end{theorem}




\subsection{Technical overview}\label{sec:overview_multiway}

Here, we introduce the technical ingredients that underpin our results on private multiway cut and discuss the key concepts behind our analysis of privacy and utility guarantees. Under edge-level privacy, there might be a difference in the edge weights between any pair of vertices. A trivial idea is to add noise between each pair of vertices, which incurs an additive error of $\tilde{\Theta}(n^{1.5})$ for privately solving optimization problems related to cuts, which is usually far from optimal. In particular, the non-efficient exponential mechanism yields additive error of only $O(n \log k)$ (and we provide a matching lower bound Section~\ref{sec:lowerbound_multiway}). To get a solution efficiently, the elegant work by \cite{DBLP:conf/nips/DalirrooyfardMN23} proposes a private algorithm which builds on a well-known 2-approximation algorithm for multiway cut. This approach involves finding the min $s$-$t$ cut at most $O(\log k)$ times. As a result, their algorithm naturally incurs a multiplicative error of $2$ but achieves a nearly tight additive error of $O(n\log^2 k)$.

To break the barrier of 2-approximation for efficient and private multiway cut, the most natural idea is to first relax the problem to finding the optimal \textit{simplex embedding}, which can be formulated as a linear program~\cite{DBLP:journals/jcss/CalinescuKR00}. After this, a randomized rounding procedure is applied as post-processing to obtain an integral solution for the multiway cut. In particular, for any $k\in \mathbb{N}_+$, we define the $k$-simplex to be $\Delta_k := \{x\in \mathbb{R}^k_{\geq 0}: \sum_{i\in [k]} x^{i} = 1\}$. Then, given a weighted graph $G = ([n], E, w)$ and a set of terminals $T\subseteq [n]$ where $T = \{s_1, s_2, \cdots, s_k\}$, we consider the following linear program for embedding vertices into the $k$-simplex (see~\cite{DBLP:journals/jcss/CalinescuKR00} for more discussion):
  \begin{gather*}
    \min_{x\in (\mathbb{R}^k)^n}\quad \frac{1}{2}\sum_{e = \{u,v\}\in E} w(e)\|x_u - x_v\|_1. \\
    \begin{aligned}
    \textup{s.t.}\quad &x_{s_i}  = \mathbf{e}_i, &\forall i\in [k] \\
                       &x_u \in \Delta_k , &\forall u\in [n]. \\
    \end{aligned}
    \end{gather*}
Here, each vertex $u\in [n]$ is attached with a $k$-dimensional vector in $\Delta_k$, and for each terminal $s\in T$, $x_{s}$ is fixed as the $i$-th canonical basis vector $\mathbf{e}_i$. Then, if for every non-terminal vertex $u\in [n]\backslash T$, $x_u$ chooses the integral solution, i.e., $x_u$ is one of $\{\mathbf{e}_1, \mathbf{e}_2, \cdots, \mathbf{e}_k\}$, it is easy to verify that the value of the objective is the size of multiway cut specified by such integral solutions. This relaxation has been extensively studied~\cite{cualinescu1998improved, karger2004rounding, buchbinder2013simplex, sharma2014multiway}, and the best-known rounding scheme yields approximation ratio of approximately $1.3$~\cite{buchbinder2013simplex, sharma2014multiway}.

Therefore, the main step is to privately solve the linear program for simplex embedding while minimizing the additive error. As discussed earlier, applying the results in \cite{hsu2014privately} for general private LP solving incurs an error of $\tilde{O}(n^{1.5})$. 
To achieve a linear dependency on $n$, inspired by the shifting mechanism in Section \ref{sec:shifting}, our approach is to add {Laplace} noise {\em only} between each terminal and non-terminal vertex to achieve pure differential privacy
\footnote{In \cite{DBLP:conf/nips/DalirrooyfardMN23}, to privately find the min $s$-$t$ cut, the authors also add noise from an {\em exponential distribution} between the two terminals $s$, $t$ and all other vertices. However, since the exact minimum  $s$-$t$ cut can be found in polynomial time, they do not consider using a linear program, which leads to a significantly different analysis.}
, as we will show in Section \ref{sec:privacy_analysis} that these edges actually form a dominating set for the fractional multiway cut (i.e., optimal simplex embedding) with a sensitivity of $O(k)$.
This allows us to compare the probabilities of being selected as optimal for each possible fractional solution $x\in (\Delta_k)^n$ between a pair of neighboring graphs with Laplace noises partially added.
Given that the sensitivity of the dominating set for the fractional multiway cut is $O(k)$, preserving $(\epsilon,0)$-differential privacy through the shifting mechanism requires setting the standard deviation of each noise to at least $O( {k}/\epsilon)$. Consequently, the expected amount of the largest noise is $\tilde{\Theta}(k)$. Given the infinite number of fractional cuts (which precludes the use of the union bound) and the fact that there are at most $O(nk)$ noises on each cut, the cumulative error across each cut will be at most 
 $\tilde{O}(nk^2)$. To save an extra $k$ factor and obtain the $\tilde{O}(n{k})$ upper bound, our observation is that since we are adding noises between terminals and non-terminals, then the noises are very ``sparse'' on each \textit{uncut}. To better understand this, consider a star graph where there are $k$ terminals and only one non-terminal $u$. Then there are $k$ uncuts and each uncut is just one edge connecting $u$ and each terminal, containing one noise. However, each cut contains $k-1$ noises. In the case of the star graph, choosing a minimum  multiway cut is equivalent to selecting the largest edge, that is, maximizing the uncut, which only incurs $O(\log k)$ error. In section \ref{sec:utility_multiway}, we extend this intuition to general graphs, and derive the $\tilde{O}(n{k})$ error bound by actually analyzing the utility loss in maximizing uncut rather than minimizing the multiway cut.

\subsection{The algorithm for private simplex embedding}\label{sec:alg_simplex_embedding}

Given an integer $k\in \mathbb{N}_+$, we define the $k$-simplex to be $\Delta_k = \{x\in \mathbb{R}^k: \sum_{i=1}^k x^{(i)} = 1 \land x^{(i)} \geq 0, \forall i\}$. Suppose there are $n\geq k$ points in the $k$-simplex $\Delta_k$, and exactly $k$ points among them are known as ``terminals'', where the $i$-th $(1\leq i\leq k)$ terminal has value $1$ at its $i$-th coordinate and $0$ at all other terminals. That is, if $T = \{s_1,s_2,\cdots, s_k\}$ is the collection of terminals, then $s_i = \mathbf{e}_i$. For each point $u\in [n]$, let $x_u\in \Delta_k$ be the position of $u$. For any two different points in $[n]$, we define a \textit{non-negative} and \textit{symmetric} cost function $c:[n]\times [n] \rightarrow \mathbb{R}_{\geq 0}$. For each pair $u,v\in [n]$, the cost between them is defined as $c(u,v)\|x_u - x_v\|_1$, which decreases as $x_u$ and $x_v$ approach each other. 

The problem of simplex embedding is to find the optimal placement of the $n-k$ non-terminal points in $\Delta_k$ in order to minimize the total sum of costs. This problem can be formalized by the following linear program \textbf{LP0}:
\begin{gather*}
    \min_{x\in (\mathbb{R}^k)^n}  \sum_{\{u, v\} \atop u\neq v} c(u,v)\|x_u - x_v\|_1 \\
    \begin{aligned}
    \textup{s.t.}\quad &x_{s_i}  = \mathbf{e}_i, &\forall i\in [k] \\
                       &x_u \in \Delta_k , &\forall u\in [n]. \\
    \end{aligned}
\end{gather*}
We note that the absolute values in the objective can be eliminated using standard techniques when $c:[n]\times [n] \rightarrow \mathbb{R}$ is non-negative, thus \textbf{LP0} is a linear program.

For the privacy notion, we define two cost functions $c,c'$ as neighboring if there is at most one pair of points $u,v\in [n]$ such that $|c(u,v) - c'(u,v)|\leq 1$. When considering points as the vertices of an undirected graph and costs as the edge weights, this definition is exactly the standard notion known as \textit{edge-level} differential privacy~\cite{gupta2012iterative, blocki2012johnson, eliavs2020differentially, liu2024optimal}. Given $k, n\in \mathbb{N}_+ (k\leq n)$, a set of terminals $T\subseteq [n]$, the cost function $c$ and privacy budget $\epsilon>0$, we give the following mechanism for private simplex embedding that runs in polynomial time, which is a variant of Algorithm \ref{alg:shifting}.

\begin{center}
    \begin{tcolorbox}[=sharpish corners, colback=white, width=1\linewidth]\label{alg:simplex_embedding}

        \begin{center}
      \textbf{\emph{Algorithm \ref*{alg:simplex_embedding}: Simplex Embedding with Pure-DP}}
        \end{center}
        \vspace{6pt}
    \begin{enumerate}
  \item \textbf{Generating Noise:} 
   For each terminal $t\in T$ and each non-terminal $u\in [n]\backslash T$, sample $Z_{\{t,u\}}\sim \mathsf{Lap}(b)$ independently. Here, $b = \frac{\sqrt{2}k}{\epsilon}$.

  \item \textbf{Solving LP:} Run the solver for the following linear program \textbf{LP1}:
  \begin{gather*}
     \min_{x\in (\mathbb{R}^k)^n}\quad g(x) = \sum_{\{u, v\} \atop u\neq v} c(u,v)\|x_u - x_v\|_1 + \sum_{t\in T}\sum_{u\in [n]\backslash T} Z_{\{t,u\}}\|x_t - x_u\|_1. \\
    \begin{aligned}
    \textup{s.t.}\quad &x_{s_i}  = \mathbf{e}_i, &\forall i\in [k] \\
                       &x_u \in \Delta_k , &\forall u\in [n]. \\
    \end{aligned}
    \end{gather*}

\end{enumerate}
  \end{tcolorbox} 
  \end{center}

We remark that the solver in the second step of Algorithm \ref{alg:simplex_embedding} will output an optimal solution in polynomial time if $c(u,v)\geq 0$ for all $u,v\in [n]$:
\begin{remark}\label{re:remove_abs}
  Notice that for any terminal $t\in T$ and non-terminal $u \in [n]\backslash T$, $$\|x_t - x_u\|_1 = 2(1-x_u^{(t)}) \geq 0$$ where $x_u^{(t)}$ is the $t$-th entry of $x_u$. Thus, the optimization problem in step 2 is indeed a linear program and can be solved in polynomial time.
\end{remark}

\subsection{Analysis of private simplex embedding with the shifting mechanism}
As opposed to trivially adding noise to the cost between \textit{any} pair of vertices, our Algorithm \ref{alg:simplex_embedding} only adds noise between terminals and each non-terminal point. We will show later that this approach dramatically reduces the scale of noises, but also makes the privacy analysis highly non-trivial. Below, we give the privacy and utility guarantees of Algorithm \ref{alg:simplex_embedding}. 

\subsubsection{Privacy analysis}\label{sec:privacy_analysis}
In this section, we provide a proof of the following privacy guarantee for Algorithm \ref{alg:simplex_embedding} in terms of the privacy notion defined in Section \ref{sec:alg_simplex_embedding}:

\begin{theorem}\label{thm:privacy_of_simplex_embedding}
    Fix any $\epsilon>0$. Algorithm \ref{alg:simplex_embedding} preserves $(\epsilon,0)$-differential privacy.
\end{theorem}

In the following discussion, without loss of generality, we assume that the fractional optimal solution of \textbf{LP1} in step 2 is \textit{unique}. If this is not the case, we define a tie-breaking rule that selects the optimal solution closest to the origin with respect to the $\ell_2$ distance.\footnote{In particular, it has been proved in \cite{mangasarian1984normal} that by adding a quadratic perturbation $\eta x^\top x$ to the linear objective $c^\top x$ for a sufficiently small $\eta$, the perturbed problem solves the original problem. Further, the optimal solution of the program would be the unique $2$-norm projection of the origin on the original optimal solution set.} Indeed, it is also well known that by adding small random perturbations to the objective, the optimal solution of \textbf{LP1} will almost surely be unique~\cite{mangasarian1984normal, spielman2004smoothed}, since an LP has infinitely many optimal solutions only if two basic solutions attain the same value.

The simplest cases are that the difference lies between two terminals, or between a terminal and a non-terminal. In both cases, it is straightforward to verify that outputting the optimal solution in step 2 is $(\epsilon,0)$-differentially private. (In particular, changing the cost between any two terminals does not affect the optimal solution set at all.)

Now, suppose $c_1$ and $c_2$ are a pair of neighboring cost functions that differ in cost for any two non-terminals $u$ and $v$. We separate all $k(n-k)$ random Laplace noises added in step 1 into two categories $\mathcal{Z}_{uv}$ and $\mathcal{Z}\backslash \mathcal{Z}_{uv}$, where $\mathcal{Z}_{uv} = \{Z_{\{t,x\}}: x\in \{u,v\}, t\in T\}$ is the random noises added between terminal and $u$ (or $v$), and $\mathcal{Z}$ contains all the $k(n-k)$ noises. Next, we use a coupling trick. Consider the following protocol of sampling noise in $\mathcal{Z}_{uv}$ and $\mathcal{Z}\backslash \mathcal{Z}_{uv}$ for both $c_1$ and $c_2$:

\begin{enumerate}
  \item Sample $k(n-k-2)$ random Laplace variables in $\mathcal{Z}\backslash \mathcal{Z}_{uv}$ independently.
  \item Add the \textbf{same} noise on both $c_1$ and $c_2$ according to $\mathcal{Z}\backslash \mathcal{Z}_{uv}$.
  \item Sample $2k$ i.i.d. random Laplace variables in $\mathcal{Z}_{uv}$, and add noise on $c_1$ according to $\mathcal{Z}_{uv}$.
  \item \textbf{Resample} $2k$ i.i.d. random Laplace variables in $\mathcal{Z}_{uv}$, and add noise on $c_2$ according to $\mathcal{Z}_{uv}'$.
\end{enumerate}
In the procedure above, the random costs added to this pair of neighboring datasets partly shares the same randomness,  specifically in $\mathcal{Z}\backslash \mathcal{Z}_{uv}$. In the meanwhile, the distribution of the random noises added to each cost function follows precisely the same distribution as described in the algorithm. We remark that this procedure is used solely for analysis and does not occur during the actual execution of the algorithm. 

Suppose every noise in $\mathcal{Z}\backslash \mathcal{Z}_{uv}$ is decided. We observe that if $x_u$ and $x_v$ are fixed, then changing the value of the cost between $u$ and $v$ or sample noise in $\mathcal{Z}_{uv}$ will change the \textbf{same} value in the objective for all feasible solutions (with the same fixed $x_u$ and $x_v$), since once $x_u$ and $x_v$ are settled, then the value of all $\|x_a-x_t\|_1$ where $a\in \{u,v\}$ and $t\in T$ are all fixed as constants. Therefore, the optimal solution does not change. Given this fact, we let $x^*(x_u,x_v)$ be the optimal solution of $\textbf{LP1}$ for some fixed $x_u$ and $x_v$, and let $$\mathcal{X}^* = \{x^*(x_u,x_v), x_u,x_v\in \Delta_k\}$$ be the set of optimal solutions for both the simplex embedding problem parameterized by $c_1$ and $c_2$, given all possible pairs $x_u$ and $x_v$. It is easy to verify that the global optimal solution for either $c_1$ or $c_2$ lies within $\mathcal{X}^*$.

Next, we consider the probability ratio that the \textbf{LP1}s for both $c_1$ and $c_2$ choose the exact same unique optimal solution $x^*\in \mathcal{X}^*$, under the coupling strategy mentioned above. Let 
\begin{equation*}
  \begin{aligned}
    g(x) = \sum_{\{a,b\} \atop a\neq b}c(a,b)\|x_a -x_b\|_1 &+ \sum_{t\in T}\sum_{w\in [n]\backslash (\{u,v\}\cup T)} Z_{\{t,w\}}\|x_w-x_t\|_1 \\
    & +\sum_{t\in T} Z_{\{t,u\}}\|x_u-x_t\|_1 + \sum_{t\in T} Z_{\{t,v\}}\|x_v-x_t\|_1
  \end{aligned}
\end{equation*} 
and 
\begin{equation*}
  \begin{aligned}
    g'(x) = &\sum_{ \{a,b\} \neq \{u,v\} \atop a\neq b}c(a,b)\|x_a -x_b\|_1 + \sum_{t\in T}\sum_{w\in [n]\backslash (\{u,v\}\cup T)} Z_{\{t,w\}}\|x_w-x_t\|_1 \\
    & +(c(u,v) + \Delta) \|x_u-x_v\|_1 + \sum_{t\in T} Z_{\{t,u\}}\|x_u-x_t\|_1 + \sum_{t\in T} Z_{\{t,v\}}\|x_v-x_t\|_1
  \end{aligned}
\end{equation*} 
be the objective function of \textbf{LP1} in terms of $c_1$ and $c_2$ respectively. Here, $\max\{-1, -c(\{u,v\})\} \leq \Delta \leq 1$ is the difference. In the following argument we always assume $\Delta = 1$, since the proof for all possible $\Delta$ simply follows from the case where $\Delta = 1$. 

Suppose all noises in $\mathcal{Z}\backslash \mathcal{Z}_{uv}$ have been decided, say $\mathcal{Z}\backslash \mathcal{Z}_{uv} = s\in \mathbb{R}^{(n-k-2)k}$. Fixing any $x_u, x_v\in \Delta_k$, let $S\subseteq \mathbb{R}^{2k}$ be the range in $\mathbb{R}^{2k}$ such that 
 $$\mathcal{Z}_{uv}\in S(s) \Leftrightarrow \min \{g(x): x\in \mathcal{X}^*\} = x^*(x_u,x_v),$$
and 
$$\mathcal{Z}_{uv}\in S'(s) \Leftrightarrow \min \{g'(x): x\in \mathcal{X}^*\} = x^*(x_u,x_v).$$
We will abbreviate $S(s)$ (and $S'(s)$) as $S$ (and $S'$) if we do not have to specify $s$. Since both $\mathbf{LP1}$s for $c_1$ and $c_2$ have a unique optimal solution, then $S(s)$ and $S'(s)$ are well-defined for any $s\in \mathbb{R}^{(n-2-k)k}$. Then, we have the following lemma:

\begin{lemma}\label{lem:reduction_to_uv}
    Fix any $\epsilon>0$. If, for any fixed $\mathcal{Z}\backslash \mathcal{Z}_{uv} = s\in \mathbb{R}^{(n-k-2)k}$ we have 
    $$\Pr{\mathcal{Z}_{uv}\in S(s)} \leq e^\epsilon \Pr{\mathcal{Z}_{uv}\in S'(s)},$$
    then Algorithm \ref{alg:simplex_embedding} preserves $(\epsilon,0)$-differential privacy.
\end{lemma}
\begin{proof}
    Notice that Algorithm \ref{alg:simplex_embedding} outputs a solution for \textbf{LP1}. To show that Algorithm \ref{alg:simplex_embedding} is $(\epsilon,0)$-differentially private, it is equivalent to showing that  for an arbitrary $x^*(x_u,x_v)\in \mathcal{X}^*$,
\begin{equation}\label{eq:expression}
    \Pr{g(x) \text{ is minimized by } x^*(x_u,x_v)} \leq e^\epsilon \cdot \Pr{g'(x) \text{ is minimized by } x^*(x_u,x_v)} .
\end{equation}
We note that 
$$\Pr{g(x) \text{ is minimized by } x^*(x_u,x_v)} = \int_{s\in \mathbb{R}^{(n-2-k)k}} \text{pdf}_{\mathcal{Z}\backslash \mathcal{Z}_{uv}}(s) \Pr{\mathcal{Z}_{uv}\in S(s)} d s$$
and 
$$\Pr{g'(x) \text{ is minimized by } x^*(x_u,x_v)} = \int_{s\in \mathbb{R}^{(n-2-k)k}} \text{pdf}_{\mathcal{Z}\backslash \mathcal{Z}_{uv}}(s) \Pr{\mathcal{Z}_{uv}\in S'(s)} d s.$$
Therefore, by the coupling argument and the fact that all Laplace noises in $\mathcal{Z}$ are i.i.d. sampled, \cref{eq:expression} holds when 
$$\Pr{\mathcal{Z}_{uv}\in S} \leq e^\epsilon \Pr{\mathcal{Z}_{uv}\in S'},$$
which completes the proof of Lemma \ref{lem:reduction_to_uv}.
\end{proof}

 In the inspiration of Lemma \ref{lem:reduction_to_uv}, we only need to find the proper variance of Laplace random variables in $Z_{uv}$ such that $\Pr{\mathcal{Z}_{uv}\in S} \leq e^\epsilon \Pr{\mathcal{Z}_{uv}\in S'}$ for all $\epsilon > 0$. However, we cannot directly compare $S$ and $S'$ since they may have complicated structures (after all noises in $\mathcal{Z}\backslash \mathcal{Z}_{uv}$ have been decided), but we can still analyze the ratio between the probabilities $\Pr{\mathcal{Z}_{uv}\in S}$ and $\Pr{\mathcal{Z}_{uv}\in S'}$ by the following lemmas and facts. The first lemma states that the sum of distances between any point and all terminals is invariant.

\begin{lemma}\label{lem:sum_invar}
  For any $x\in \Delta_k$ and $t\in T$,
  $$\frac{1}{2}\sum_{t\in T} {\|x - x_t\|_1} = k-1.$$
\end{lemma}
\begin{proof}
  For any $x\in \Delta_k$, we have that $\sum_{t\in T} x^{(t)} = 1$. Thus,
  \begin{equation*}
    \frac{1}{2}\sum_{t\in T}\|x -x_t\|_1  =\sum_{t\in T} |1-x^{(t)}| = k - \sum_{t\in T}x^{(t)} = k-1,
  \end{equation*}
  which completes the proof.
\end{proof}

The following lemma indicates the existence of the dominating set of multiway cut of size $O(nk)$ with sensitivity $O(k)$.

\begin{lemma}\label{lem:minimizes_two_vertices_case}
  Fix any integer $k\geq 2$. Let $\{d_{u,i}\}_{i\in [k]}$ and $\{d_{v,i}\}_{i\in [k]}$ be two sets of variables such that
  \begin{enumerate}
    \item For any $w\in \{u,v\}$ and $i\in [k]$: $0\leq d_{w,i} \leq 1$.
    \item For any $w\in \{u,v\}$, $\sum_{i\in [k]} d_{w,k} = k-1$.
  \end{enumerate}
Then, for any fixed $\{d_{u,i}^*\}_{i\in [k]}$ and $\{d_{v,i}^*\}_{i\in [k]}$ that satisfy both conditions, there always exist $a_u, a_v \in \mathbb{R}^{k}$ with $\|a_u\|_2 \leq \sqrt{k/2}$ and $\|a_v\|_2 \leq \sqrt{k/2}$ such that 
\begin{equation}\label{eq:minimizes_two_vertices_case}
  \sum_{i\in [k]} a_{u,i}d_{u,i} + \sum_{i\in [k]} a_{v,i}d_{v,i} + \frac{1}{2}\sum_{i\in [k]} |d_{u,i} - d_{v,i}|
\end{equation}
achieves minimum value on $\{d_{u,i}^*\}_{i\in [k]}$ and $\{d_{v,i}^*\}_{i\in [k]}$.
\end{lemma}

\begin{proof}
  Before getting into the proof, we note that Lemma \ref{lem:minimizes_two_vertices_case} does \textit{not} ask that $\{d_{u,i}^*\}_{i\in [k]}$ and $\{d_{v,i}^*\}_{i\in [k]}$ are the \textit{only} minimum choices. Let $U\subseteq [k]$ be the set of coordinates that $d^*_{u,i} > d^*_{v,i}$ for any $i\in U$ and $\overline{U} = [k] \backslash U$ be the set of other coordinates. Without loss of generality, we assume $|U| \leq |\overline{U}|$. (Since the subscript $u$ and $v$ are essentially symmetric, if $|U| > |\overline{U}|$, we replace $u$ and $v$ for all variables.) We construct the correction vectors $a_u$ and $a_v$ as follows:
  \begin{enumerate}
    \item For each $i\in U$, let $a_{u,i} = -1$ and $a_{v,i} = 1$.
    \item For each $i\in \overline{U}$, let $a_{u,i} = a_{v,i} = 0$.
  \end{enumerate}
  Since $|U|\leq k/2$, then clearly with such construction we have $\|a_u\|_2 = \|a_v\|_2 \leq \sqrt{k/2}$. 
  
  First, it is easy to verify that with this assignment on the coefficients, the minimum value of \cref{eq:minimizes_two_vertices_case} is at least $0$. Because $\sum_{i\in [k]} d_{u,k} = \sum_{i\in [k]} d_{v,k} = k-1$, then for any $T\subseteq [k]$,
  $$\sum_{i \in T} d_{u,i} - d_{v,i} = \sum_{i\in [k]\backslash T} d_{v,i} - d_{u,i}.$$
  Thus, we see that
  $$\sum_{i \in U} |d_{u,i} - d_{v,i}| = \sum_{i\in \overline{U}} |d_{u,i} - d_{v,i}| = \frac{1}{2}\sum_{i\in [k]} |d_{u,i} - d_{v,i}|.$$
  Therefore, 
  \begin{equation*}
    \begin{aligned}
      \sum_{i\in [k]} a_{u,i}d_{u,i} + \sum_{i\in [k]} a_{v,i}d_{v,i} + \frac{1}{2}\sum_{i\in [k]} |d_{u,i} - d_{v,i}| &= -\sum_{i\in U} d_{u,i} + \sum_{i\in U} d_{v,i} + \frac{1}{2}\sum_{i\in [k]} |d_{u,i} - d_{v,i}|\\
      &\geq -\sum_{i\in U} |d_{u,i} - d_{v,i}| + \frac{1}{2}\sum_{i\in[k]}|d_{u,i} - d_{v,i}|\\
      & = -\sum_{i\in U} |d_{u,i} - d_{v,i}| + \sum_{i\in U}|d_{u,i} - d_{v,i}| = 0.
    \end{aligned}
  \end{equation*}
  Next, we show that \cref{eq:minimizes_two_vertices_case} achieves $0$ on all $\{d_{u,i}\}$'s and $\{d_{v,i}\}$'s as long as $\{i: d_{u,i} > d_{v,i}\}$ is exactly $U$, which includes $\{d_{u,i}^*\}_{i\in [k]}$ and $\{d_{v,i}^*\}_{i\in [k]}$. 
  Indeed, if $\{d_{u,i}\}$ and $\{d_{v,i}\}$ satisfy $\{i: d_{u,i} > d_{v,i}\} = U$, then 
  \begin{equation*}
    \begin{aligned}
      \sum_{i\in [k]} |d_{u,i} - d_{v,i}| = \sum_{i \in U} d_{u,i} - d_{v,i} + \sum_{i\in \overline{U}} d_{v,i} - d_{u,i} = 2\sum_{i\in U} (d_{u,i} - d_{v,i}).
    \end{aligned}
  \end{equation*}
  Thus, 
  \begin{equation*}
    \begin{aligned}
      \sum_{i\in [k]} a_{u,i}d_{u,i} + \sum_{i\in [k]} a_{v,i}d_{v,i} + \frac{1}{2}\sum_{i\in [k]} |d_{u,i} - d_{v,i}| &= -\sum_{i\in U} d_{u,i} + \sum_{i\in U} d_{v,i} + \sum_{i\in U} (d_{u,i} - d_{v,i}) =0.
    \end{aligned}
  \end{equation*}
\end{proof}

\begin{fact}\label{fac:common_minimum}
  Let $a:\mathcal{R}\rightarrow \mathbb{R}$ and $b:\mathcal{R}\rightarrow \mathbb{R}$ be two real-valued function on some range $\mathcal{R}$. Suppose $x^* \in \mathcal{R}$ is the unique solution that minimizes $a(x)$ and $x^*$ is also one of the solutions that minimize $b(x)$, then $x^*$ is the unique solution that minimizes $c(x) = a(x) + b(x)$.
\end{fact}
\begin{proof}
  Suppose $x\neq x^*$; then $a(x)>a(x^*)$ and $b(x)\geq b(x^*)$, which is saying that $c(x)>c(x^*)$ and thus $x^*$ is the only minimizer, completing the proof.
\end{proof}

\noindent Let 
$$h(x) = \sum_{\{a,b\}\in E }c(\{a,b\})\|x_a -x_b\|_1 + \sum_{t\in T}\sum_{w\in [n]\backslash (\{u,v\}\cup T)} Z_{\{t,w\}}\|x_w-x_t\|_1  ,$$
then we rewrite $g(x)$ and $g'(x)$ as 
$$g(x) = h(x) + \sum_{t\in T} Z_{\{t,u\}}\|x_u-x_t\|_1 + \sum_{t\in T} Z_{\{t,v\}}\|x_v-x_t\|_1$$
and 
$$g'(x) = g(x) +  \|x_u-x_v\|_1.$$


\noindent In the context of Fact \ref{fac:common_minimum}, let $a(x)$ be $g(x)$ and let
\begin{equation*}
    \begin{aligned}
        b(x) :&=\|x_u-x_v\|_1 + \sum_{i\in [k]} a_{u,i} \|x_u-x_t\|_1 +\sum_{i\in [k]} a_{v,i} \|x_v-x_t\|_1\\
        & = \sum_{i\in [k]} |(1-x_u^{(i)}) - (1-x_v^{(i)})| + \sum_{i\in [k]} a_{u,i} \|x_u-x_t\|_1 +\sum_{i\in [k]} a_{v,i} \|x_v-x_t\|_1\\
        & = \sum_{i\in [k]} \frac{|\|x_u - x_t\|_1 - \|x_v - x_t\|_1|}{2} + \sum_{i\in [k]} a_{u,i} \|x_u-x_t\|_1 +\sum_{i\in [k]} a_{v,i} \|x_v-x_t\|_1.
    \end{aligned}
\end{equation*}
Then, we can just let $d_{w,t} = \frac{\|x_w - x_t\|_1}{2}$ for any $w\in \{u,v\}$ and $t\in [k]$. By Lemma \ref{lem:sum_invar}, such assignment satisfies both conditions in Lemma \ref{lem:minimizes_two_vertices_case}. Therefore, by applying Lemma \ref{lem:minimizes_two_vertices_case} together with Fact \ref{fac:common_minimum}, we conclude that if the objective $g(x)$ is minimized by any fixed solution $x^*(x_u,x_v)\in \mathcal{X}^*$, then there exist $2k$ correction coefficients $\{a_u,a_v\}$ with $\|a_u\|^2_2 + \|a_v\|_2^2 = k$ such that 
 \begin{equation*}
  \begin{aligned}
  &g(x) + \|x_u-x_v\|_1 + \sum_{i\in [k]} a_{u,i} \|x_u-x_t\|_1 +\sum_{i\in [k]} a_{v,i} \|x_v-x_t\|_1 \\
  & = h(x) + \|x_u-x_v\|_1 + \sum_{t\in T} (Z_{\{t,u\}} + a_{u,t})\|x_u-x_t\|_1 + \sum_{t\in T} (Z_{\{t,v\}}+a_{v,t})\|x_v-x_t\|_1
\end{aligned}
\end{equation*}
is minimized by the unique minimizer $x^*(x_u,x_v)$. Then by the definition of $S$ and $S'$, we have 
$$\mathcal{Z}_{uv} \in S \Rightarrow \mathcal{Z}_{uv} +  a_{uv} \in S',$$
where $a_{uv}\in \mathbb{R}^{2k}$ is the concatenation of $a_u$ and $a_v$. This is saying that $S + a_{uv} \subseteq S'$, and thus
\begin{equation}\label{eq:eq2}
  \Pr{\mathcal{Z}_{uv} \in S'} \geq \Pr{\mathcal{Z}_{uv} \in S + a_{uv}}.
\end{equation}
Next, to compare the probability of $\mathcal{Z}_{uv} \in S$ and $\mathcal{Z}_{uv} \in S + a_{uv}$, we apply the basic Laplace mechanism. In the context of Lemma \ref{lem:laplace}, if we set $R = S$, $v = \mathbf{0}_{2k}$ and $v' = -a_{uv}$, and considering that we have $\|a_{uv}\|_1 \leq \sqrt{2k}\|a_{uv}\|_2 \leq \sqrt{2}k$ from Lemma \ref{lem:minimizes_two_vertices_case}, it is clear that if $Z_{uv}$ are i.i.d Laplace random variables from $\mathsf{Lap}(b)$ with $b = \frac{\sqrt{2}k}{\epsilon}$, 
\begin{equation}\label{eq:bounded_uv}
  \begin{aligned}
    \Pr{\mathcal{Z}_{uv}\in S} \leq e^\epsilon \Pr{\mathcal{Z}_{uv} -a_{uv}\in S} &= e^\epsilon \Pr{\mathcal{Z}_{uv} \in  S + a_{uv}} \\
    &\leq e^\epsilon \Pr{\mathcal{Z}_{uv} \in S'}.
  \end{aligned}
\end{equation}
Here, the first inequality comes from the Laplace mechanism and the last inequality comes from \cref{eq:eq2}. Thus, \cref{eq:bounded_uv} together with Lemma \ref{lem:reduction_to_uv} completes the proof of Theorem \ref{thm:privacy_of_simplex_embedding}.

\begin{remark}
  Again, we note that the range $S$ \textbf{does not} depend on the randomness of noise, instead it only depends on the values of the input cost function. This is why we could let $R$ in Lemma \ref{lem:laplace} be just $S$ and directly apply the Laplace mechanism.
\end{remark}

\subsubsection{Utility analysis}\label{sec:utility_multiway}

Let the non-negative and symmetric function $c = [n]\times [n]\rightarrow \mathbb{R}_{\geq 0}$ be the input cost function. We define the total sum of costs as $W := 2\sum_{\{u,v\}} c(u,v)$, in which the cost between each pair of points is counted twice. Given any fractional placement of the points $x\in (\Delta_k)^n$, we define $\mathcal{E}(x)$ be the size of the total cost of such placement, that is:
$$\mathcal{E}(x) = \sum_{ \{u,v\}\atop u\neq v} c(u,v){\|x_u - x_v\|_1}.$$
We give the following of the utility guarantee of Algorithm \ref{alg:simplex_embedding}:

\begin{theorem}\label{thm:utility_simplex_embedding}
    Fix any $\epsilon>0$. For any given $k,n\in \mathbb{N}_+ (k\leq n)$, set of terminals $T\subseteq[n]$ and cost function $c$, Algorithm \ref{alg:simplex_embedding} outputs a placement $\hat{x}^* \in (\Delta_k)^n$ such that 
    $$\mathbb{E}\left[\mathcal{E}(\hat{x}^*) \right] - \mathsf{OPT}(c,n,k) \leq O\left(\frac{n{k}\log (k)}{\epsilon}\right).$$
    Here, $\mathsf{OPT}(c,n,k)$ is the optimal value of the simplex embedding problem defined by $c,n$ and $T$. Further, we have that with high probability,
       $\mathcal{E}(\hat{x}^*) - \mathsf{OPT}(c,n,k) \leq O\left(\frac{n{k}\log (nk)}{\epsilon}\right).$
\end{theorem}

 If we consider $c$ as the weights of edges of an undirected graph, then it can be verified that $\mathcal{E}(x)$ is equivalent to the size of the fractional cut specified by the placement $x$ (up to of coefficient $1/2$). We denote by $\mathcal{U}(x)$ the size of ``uncut'' specified by $x$, that is, $\mathcal{U}(x) = W - \mathcal{E}(x)$ for any $x\in (\Delta_k)^n$.
We first consider the following linear program $\mathbf{LP2}$:

\begin{gather*}
  \min_{x\in (\mathbb{R}^k)^n}\sum_{ \{u,v\}\atop u\neq v} c(u,v)(2-\|x_u - x_v\|_1)  \\
  \begin{aligned}
  \textup{s.t.}\quad &x_{s_i}  = \mathbf{e}_i, &\forall i\in [k] \\
                     &x_u \in \Delta_k , &\forall u\in [n]. \\
  \end{aligned}
  \end{gather*}
Clearly, since (with a slight abuse of notation),
$$\sum_{e = \{u,v\}}c(e)(2-\|x_u - x_v\|_1) + \sum_{e = \{u,v\}} c(e)\|x_u  -x_v\|_1 = 2\sum_{\{u,v\}}c(u,v) = W,$$
then $\mathcal{U}(x) = \sum_{ \{u,v\}} c(u,v)(2-\|x_u - x_v\|_1)$. Therefore, solving $\mathbf{LP2}$ is equivalent to maximizing uncut $\mathcal{U}(x)$. 

Similarly, suppose all noises in $\mathcal{Z}$ are fixed as $Z_{\{t,u\}} = z_{t,u}\in \mathbb{R}$, and let $\hat{c}$ be the noisy costs (whose cost between any pair of terminal $t$ and non-terminal $u$ has an extra $z_{t,u}$ additive noise). We write $\widehat{\mathcal{E}}(x)$ and $\widehat{\mathcal{U}}(x)$ as the size of cut and uncut with the noisy costs specified by $x$ respectively. Clearly, the total sum of costs in $\hat{c}$ is 
$$\widehat{W} = W + 2\sum_{t\in T}\sum_{u\in [n]\backslash T} z_{t,u},$$
and $\widehat{W} = \widehat{\mathcal{E}}(x) +  \widehat{\mathcal{U}}(x)$. Now, we write the following linear program $\mathbf{LP3}$ which maximizes the uncut $\widehat{\mathcal{U}}(x)$:
\begin{gather*}
  \min_{x\in (\mathbb{R}^k)^n}\quad  \sum_{e = \{u,v\}} c(e)(2-\|x_u - x_v\|_1) + \sum_{t\in T}\sum_{u\in [n]\backslash T} z_{t,u}(2-\|x_t - x_u\|_1)  \\
  \begin{aligned}
  \textup{s.t.}\quad &x_{s_i}  = \mathbf{e}_i, &\forall i\in [k] \\
                     &x_u \in \Delta_k , &\forall u\in [n]. \\
  \end{aligned}
  \end{gather*}
We have the following lemma on the difference between $\mathcal{U}(x)$ and $\widehat{\mathcal{U}}(x)$.

\begin{lemma}\label{lem:little_noise_in_uncut}
  With high probability, for any fractional partition $x\in (\Delta_k)^n$, we have 
  $$|\mathcal{U}(x) - \widehat{\mathcal{U}}(x)| \leq O\left(\frac{n{k}\log (nk)}{\epsilon}\right).$$
\end{lemma}
\begin{proof}
  Suppose for any pair of terminal $t$ and non-terminal $u$, $z_{t,u}$ is independently sampled from $\mathsf{Lap}(b)$. Then by the tail bound of Laplace noise, with high probability, $\max_{t,u}|z_{t,u}|\leq O(b\log(nk))$. Suppose this holds. By lemma \ref{lem:sum_invar}, we see that for any $x_u\in \Delta_k$ (no matter whether $x_u$ is integral or not), we have
  $$\sum_{t\in T} (2-\|x_t - x_u\|_1) = 2k - (2k-2) = 2.$$
  Then, 
  $$\sum_{u\in [n]\backslash T} \sum_{t\in T}z_{t,u}(2-\|x_t-x_u\|_1)\leq  2\sum_{u\in [n]\backslash T}b{\log(nk)} \leq 2nb{\log(nk)}.$$
  Finally, letting $b = O(k/\epsilon)$ competes the proof of Lemma \ref{lem:little_noise_in_uncut}.
\end{proof}

\noindent With Lemma \ref{lem:little_noise_in_uncut}, we are now ready to state the proof of Theorem \ref{thm:utility_simplex_embedding}.

\begin{proof}
(Of Theorem \ref{thm:utility_simplex_embedding}.) In the original cost function $c$, we remark that $\mathcal{E}(x) + \mathcal{U}(x) = W $ for any fractional partition $x$, while with the noisy costs, $\widehat{\mathcal{E}}(x) + \widehat{\mathcal{U}}(x) = \widehat{W}$ for any fractional partition $x$. Suppose $x^*$ is the optimal fractional solution on simplex embedding with the original cost $c$, and $\hat{x}^*$ is the optimal fractional solution with noisy costs $\hat{c}$ (i.e., optimal solution of \textbf{LP1}). We have
\begin{equation*}
  \begin{aligned}
    & \mathcal{U}(x^*) - \mathcal{U}(\hat{x}^*) = W - \mathcal{E}(x^*) - W + \mathcal{E}(\hat{x}^*) \geq 0 \quad \quad\text{and}\\
    &  \widehat{\mathcal{U}}(\hat{x}^*) -  \widehat{\mathcal{U}}(x^*) = W' -  \widehat{\mathcal{E}}(\hat{x}^*) - W' +  \widehat{\mathcal{E}}(x^*) \geq 0.
  \end{aligned}
\end{equation*}
Therefore,
\begin{equation*}
  \begin{aligned}
    |\mathcal{E}(x^*) - \mathcal{E}(\hat{x}^*)| = |\mathcal{U}(x^*) - \mathcal{U}(\hat{x}^*)| &= \mathcal{U}(x^*) - \mathcal{U}(\hat{x}^*)\\
    & \leq \mathcal{U}(x^*) - \mathcal{U}(\hat{x}^*) + \widehat{\mathcal{U}}(x^*) - \widehat{\mathcal{U}}(\hat{x}^*) \\
    &\leq |\mathcal{U}(x^*) - \widehat{\mathcal{U}}({x}^*)| + |{\mathcal{U}}(\hat{x}^*) - \widehat{\mathcal{U}}(\hat{x}^*)| \\
    &\leq O\left(\frac{n{k}\log (nk)}{\epsilon}\right).
  \end{aligned}
\end{equation*}
This completes the proof of the utility guarantee with high probability. Here, the first inequality follows from $\widehat{\mathcal{U}}(\hat{x}^*) -  \widehat{\mathcal{U}}(x^*) \geq 0$, and the second equality follows from $\mathcal{U}(x^*) - \mathcal{U}(\hat{x}^*) \geq 0$. The final inequality follows from Lemma \ref{lem:little_noise_in_uncut}. The guarantee in expectation can be derived almost identically, except using a union bound over the $k$ noises incident on each point (instead of all $nk$ noises) and linearity of expectation.
\end{proof}


\subsection{The application for private multiway cut: Proof of Theorem \ref{thm:main_multiway_cut}}
Given an undirected and weighted graph $G = ([n],E)$, non-negative edge weights $c:{E}\rightarrow \mathbb{R}_+$, terminals $T = \{s_1,\cdots, s_k\}$, we attach each vertex in $[n]$ with a vector $x_u \in \mathbb{R}^k$. Then, the problem of finding the optimal multiway cut with respect to given terminals $T$ can be formulated as the following integer linear program \textbf{LP4}:
\begin{gather*}
    \min_{x\in (\mathbb{R}^k)^n}  \frac{1}{2}\sum_{\{u, v\} \in E} c(u,v)\|x_u - x_v\|_1 \\
    \begin{aligned}
    \textup{s.t.}\quad &x_{s_i}  = \mathbf{e}_i, &\forall i\in [k] \\
                       &x_u \in  \{0,1\}^k , &\forall u\in [n]. \\
    \end{aligned}
\end{gather*}
In the above linear program, if $u$ and $v$ are in different connected components, then $x_u\neq x_v$, and $\frac{1}{2}c(u,v)\|x_u-x_v\|_1$ is the edge weight between $u$ and $v$. Therefore, the optimum value of \textbf{LP4} is also the minimum size of multiway cut on $G$. Notice that if we relax the constraint of \textbf{LP4} from $x_i\in \{0,1\}^k$ to $x_i\in \Delta_k$, then \textbf{LP4} is equivalent to \textbf{LP0} if we consider the costs between pairs in the simplex embedding problem as the edge weights. The approximation ratio with respect to such relaxation is widely studied~\cite{cualinescu1998improved, karger2004rounding, buchbinder2013simplex, sharma2014multiway}. We apply the rounding scheme given by \cite{sharma2014multiway} for the simplex embedding linear program of multiway cut:

\begin{theorem}[\cite{sharma2014multiway}]\label{thm:rounding}
    There is a randomized algorithm $\mathcal{R}$ such that on a fractional placement $x\in (\Delta_k)^n$ where $x\in (\Delta_k)^n$ costs $\mathcal{E}(x)$ on some graph $G$, $\mathcal{R}$ outputs a solution $\tilde{x}\in (\{0,1\}^k)^n$ such that 
    $$\mathbb{E}[\mathcal{E}(\tilde{x})] \leq 1.2965 \cdot \mathcal{E}(x),$$
    where $\mathcal{E}(\tilde{x})$ is the cost of cut $\tilde{x}$ on $G$.
\end{theorem}

Since the rounding scheme $\mathcal{R}$ only requires $x\in (\Delta_k)^n$ as input, then we can treat it as a post-processing step that does not occur any loss of privacy. Putting everything all together, we give the following algorithm for private multiway cut:

\begin{center}
    \begin{tcolorbox}[=sharpish corners, colback=white, width=1\linewidth]\label{alg:multiway_cut}

        \begin{center}
      \textbf{\emph{Algorithm \ref*{alg:multiway_cut}: Private multiway cut}}
        \end{center}
        \vspace{6pt}
    \begin{enumerate}
  \item \textbf{Solving simplex embedding:} 
   Given an undirected graph $G$, a set of terminals $T$ and privacy budget $\epsilon$, run Algorithm \ref{alg:simplex_embedding} with costs $c$ as the non-negative edge weights.
   
  \item \textbf{Rounding}: Let the output of Algorithm \ref{alg:simplex_embedding} be ${x}_{pri}\in (\Delta_k)^n$, run the rounding scheme $\mathcal{R}$ in Theorem \ref{thm:rounding} and output the solution for multiway cut $\tilde{x}_{pri}\in \{1,\cdots, k\}^n$.   
\end{enumerate}
  \end{tcolorbox} 
  \end{center}

\noindent With algorithm \ref{alg:multiway_cut}, we are now ready to give the proof of Theorem \ref{thm:main_multiway_cut}:
\begin{proof}
    (Of Theorem \ref{thm:main_multiway_cut}.) The privacy guarantee of Theorem \ref{thm:main_multiway_cut} directly follows from Theorem \ref{thm:privacy_of_simplex_embedding} and the fact that differential privacy is robust to post-processing. As for the utility, given the input graph $G$, again let $\mathcal{E}(x)$ be the cost of any (fractional) solution for multiway cut on $G$. Let $\mu_G$ be the output distribution of Algorithm \ref{alg:simplex_embedding}. Notice that
    \begin{equation*}
        \begin{aligned}
            \mathbb{E}[\mathcal{E}(\tilde{x}_{pri})] &= \int_{x\in (\Delta_k)^n} \mathbb{E}[\mathcal{E}(\tilde{x}_{pri}) | x_{pri} = x] \mu_G(x) dx\\
            & \leq 1.2965\int_{x\in (\Delta_k)^n} \mathcal{E}(x) \mu_G(x) dx = 1.2965 \cdot \mathbb{E}_{\mu_G}[\mathcal{E}(x)]\\
            & \leq 1.2965\cdot \mathsf{OPT}^T(G) + O\left(\frac{n{k}\log(k)}{\epsilon}\right),
        \end{aligned}
    \end{equation*}
    where $\mathsf{OPT}^T(G)$ is the optimum value of multiway cut on $G$ in terms of the terminals $T$. Here, the first inequality comes from Theorem \ref{thm:rounding} and the last inequality comes from Theorem \ref{thm:utility_simplex_embedding}, completing the proof.
\end{proof}

\subsection{Information-theoretic lower bound}\label{sec:lowerbound_multiway}
Here, we give the lower bound on solving multiway $k$-cut with pure differential privacy for $k = O(n)$. This generalizes the lower bound in \cite{DBLP:conf/nips/DalirrooyfardMN23} in terms of $\epsilon$-DP. 
\begin{theorem}
  Fix any $\epsilon>0$, $n,k\in \mathbb{N}_+$ and $6\leq k\leq n/2$. Any $(\epsilon,0)$-differentially private algorithm on approximating multiway cut on $n$-vertex graphs has expected error at least $\frac{n\log(k/6)}{24\epsilon}$.
\end{theorem}

\begin{proof}
 Let the set of vertices $V = V' \cup T$, where $V' = \{v_1,v_2,\cdots v_{n-k}\}$ are non-terminals and $T=\{t_1,t_2,\cdots,t_k\}$ are terminals.
  For any $\tau\in [k]^{V'}$, let $G_\tau$ be the graph that is constructed by the following steps:
  \begin{enumerate}
    \item Start from an empty graph on $n$ vertices.
    \item For any $i\in V'$ and $s\in T$, if $\tau$'s $i$-th element $\tau_i = s$, then connect an edge between $v_i$ and $s$ with weight $\frac{\log (k/6)}{2\epsilon}$.
  \end{enumerate}
  Let the collection of all these graphs $\Pi = \{G_\tau |\tau\in [k]^{V'}\}$ be the input domain, and $M:\Pi \rightarrow [k]^{V'}$ be a $(\epsilon,0)$-differentially private algorithm such that for any $G\in \Pi$, $M$ outputs a solution on multiway cut. Let $\textbf{err}(M,G) = \text{size}_G(M(G)) - OPT_G$ be the gap between the output of $M$ and the optimal solution of multiway cut on $G$.  

  For any $1\leq i\leq n-k$, let $\textbf{I}_i(M,G_\tau)$ be the indicator variable such that $\textbf{I}_i(M,G_\tau) = 1$ if $M$ does not assign $v_i$ on the $\tau_i$ side, and $\textbf{I}_i(M,G_\tau) = 0$ otherwise. For each $G_\tau\in \Pi$, by the construction, we note that the optimal size of multiway cut in $G_\tau$ is zero since each pair of terminals are disconnected. Thus.
  \begin{equation}\label{eq:identity_to_error}
      \textbf{err}(M,G_\tau) = \text{size}_{G_\tau}(M(G)) = \frac{\log (k/6)}{2\epsilon}\sum_{i = 1}^{n-k} \mathbf{I}_i(M,G_{\tau}).
  \end{equation}
Next, we show that there \textit{exists} a $\tau$ such that the sum of the expected errors 
$$\sum_{i = 1}^{n-k} \mathbb{E}_{M}[\mathbf{I}_i(M,G_{\tau})]$$ is at least $\Omega(n)$. To do this, we further consider a uniformly random distribution over $\tau \in [k]^{V'}$, and it is enough to show that 

$$\mathbb{E}_{\tau} \left[\sum_{i = 1}^{n-k} \mathbb{E}_{M}[\mathbf{I}_i(M,G_{\tau})]\right] =  \sum_{i = 1}^{n-k} \mathbb{E}_{\tau}[\mathbb{E}_{M}[\mathbf{I}_i(M,G_{\tau})]] = \Omega(n)$$
for $k\leq n/2$, where the first equality comes from the linearity of expectation. Now, fix any $v_i\in V'$, and any partial configuration $\hat{\tau} \in [k]^{V'\backslash\{i\}}$.
Let $M_i: \Pi \rightarrow [k]$ be the projection of $M$ on $v_i$. Since $M$ is $\epsilon$-differentially private, by post-processing immunity, $M_i$ is also $\epsilon$-differentially private. For the fixed vertex $v_i$, let $\tau^{(1)}, \tau^{(2)}, \cdots, \tau^{(k)}$ be $k$ elements in $[k]^{V'}$ such that for any $j\in[k]$, $\tau^{(j)}$'s $i$-th element $\tau_i^{(j)} = t_j$, and for any other elements $s\in [|V'|]  \backslash \{i\}$ and any assignment $j\in [k]$, $\tau_s^{(j)} = \hat{\tau}_s^{(j)}$. In other words,  $\tau^{(1)}, \tau^{(2)}, \cdots, \tau^{(k)}$ identifies a series of $k$ graphs in which $v_i$ is connected to $k$ different terminals respectively, and all other vertices connect to the terminals that is compatible with the fixed partial configuration $\hat{\tau}$. Then, we note that 

\begin{equation}\label{eq:corrected_1}
\begin{aligned}
        \mathbb{E}_{\tau}[\mathbb{E}_{M}[\mathbf{I}_i(M,G_{\tau})]] &= \mathbb{E}_{\tau}[\mathbb{E}_{M}[\mathbf{I}_i(M_i,G_{\tau})]]\\
        & = k^{1-n}\sum_{\hat{\tau}\in [k]^{V'\backslash\{i\}}} \left(\frac{1}{k}\sum_{j\in [k]} \mathbb{E}_{M}[\mathbf{I}_i(M_i,G_{\tau}) | \tau = \tau^{(j)}] \right).
\end{aligned}
\end{equation}
For the fixed vertex $v_i\in V'$ and partial configuration $\hat{\tau}$, we define $S\subseteq [k]$ to be
\begin{equation}\label{eq:S}
    \begin{aligned}
        S := \left\{j\in[k]: \mathbb{E}_{M}[\mathbf{I}_i(M_i,G_{\tau}) | \tau = \tau^{(j)}] < \frac{1}{3}\right\}.
    \end{aligned}
\end{equation}
For the sake of contradiction, we assume $|S|\geq k/2$. For any $s, j\in S$, we also define $\mu_s(t_j)$ as the probability that on the input graph $G_{\tau^{(s)}}$, $M_i$ assigns $v_i$ on $t_j$'s side. By Markov's inequality, on any input graph $G_{\tau^{(s)}}$ where $s\in S$,

$$\mu_s(t_s)= \Pr{\mathbf{I}_i(M,G_{\tau^{(s)}}) < \frac{2}{3}} \geq \frac{1}{2}$$
since $\mathbf{I}_i(M,G_{\tau^{(s)}})$ is zero only when $M_i$ assigns $v_i$ on $t_s$'s side and otherwise $\mathbf{I}_i(M,G_{\tau^{(s)}}) = 1$. On the other hand, note that $G_{\tau^{(s)}}$ can be obtained from $G_{\tau^{(1)}}$ by editing edges with total sum of edge weights of $\frac{\log(k/6)}{\epsilon}$. Then, by the fact that $M_i$ is $\epsilon$-differentially private, 
$$\mu_1(t_s) \geq \mu_s(t_s)e^{-\epsilon \cdot \frac{\log(k/6)}{\epsilon}} = \mu_s(t_s) \cdot \frac{6}{k} \geq \frac{3}{k}$$
for any $s\in S$ by repeatedly applying the definition of differential privacy $\log(k/6)/\epsilon$ times. This means that
$$\sum_{s\in S} \mu_1(t_s) \geq |S|\cdot \frac{3}{k} \geq \frac{3}{2} > 1,$$
which contradicts the fact that $\sum_{s\in S} \mu_1(t_s) \leq \sum_{j\in [k]} \mu_1(t_j) = 1$. Therefore, we have that the size of $S$ in \cref{eq:S} is less than $k/2$, and thus we have that 

\begin{equation}\label{eq:corrected_2}
    \begin{aligned}
        \frac{1}{k}\sum_{j\in [k]} \mathbb{E}_{M}[\mathbf{I}_i(M_i,G_{\tau}) | \tau = \tau^{(j)}] \geq \frac{1}{k} \cdot \frac{1}{3}\cdot (k-|S|) \geq \frac{1}{6}.
    \end{aligned}
\end{equation}
Plug \cref{eq:corrected_2} into \cref{eq:corrected_1}, we have that for any $v_i \in V'$, 
\begin{equation}
    \begin{aligned}
         \mathbb{E}_{\tau}[\mathbb{E}_{M}[\mathbf{I}_i(M,G_{\tau})]] \geq \sum_{\hat{\tau}\in [k]^{V'\backslash\{i\}}}\frac{k^{1-n}}{6} = \frac{1}{6}.
    \end{aligned}
\end{equation}
Then, by \cref{eq:identity_to_error} and the linearity of expectation, we have that under the uniform distribution over $\tau \in [k]^{V'}$, 

$$\mathbb{E}_{\tau}\left[\mathbb{E}_M[\textbf{err}(M,G_\tau)]\right] = \frac{\log (k/6)}{2\epsilon}\sum_{i = 1}^{n-k} \mathbb{E}_{\tau}[\mathbb{E}_{M}[\mathbf{I}_i(M,G_{\tau})]] \geq (n-k)\cdot \frac{\log (k/6)}{12\epsilon},$$
which implies that there exists a $\tau \in [k]^{V'}$ such that $\mathbb{E}_M[\textbf{err}(M,G_\tau)]$ is lower bounded by $O\left(\frac{(n-k)\log (k)}{\epsilon}\right)$, completing the proof as long as $k\leq n/2$.
\end{proof}

\section{Other Applications of the Shifting Mechanism}\label{sec:app_shifting}

Here, we introduce some direct implications or applications of the framework based shifting mechanism for solving problems in private combinatorial optimization. For the needs of certain graph applications, we also give the following proposition:
\begin{proposition}\label{prop:robust_against_transit}
Fix any $S\subseteq [N]$ and $s\geq 0$. If $f:\mathcal{X}^N \rightarrow \mathcal{R}$ admits a dominating set $S$ of sensitivity $s$, then so is $f_v(x) = f(x + v)$ for any $v\in \mathcal{X}^N$ such that $x+v\in \mathcal{X}^N$.
\end{proposition}
\begin{proof}
    Fix a function $f_v$. For any pair of neighboring dataset $x,x'$, let $\hat{x} = x+v$ and $\hat{x}' = x' + v$ be another pair of neighboring datasets. Since $f$ admits a dominating set $S$ of sensitivity $s$, then there exists a vector $a\in \mathcal{X}^S$ with $\|a\|_1\leq s$ such that $a$ only depends on the value of $f(\hat{x}) = f_v(x)$ and $\hat{x} - \hat{x}' = x-x'$, and that
    \begin{align*}
       f_v({x}'|_S + a, {x}'|_{[N]\backslash S})& = 
        f(({x}' + v)|_S + a , ({x}' + v)|_{[N]\backslash S})\\
       & = f(\hat{x}'|_S + a, \hat{x}'|_{[N]\backslash S}) 
             =f(\hat{x})
        = f_v(x),
    \end{align*}
    which completes the proof.
\end{proof}

\subsection{The application of shifting mechanism in private min $s$-$t$ cut}\label{sec:app_stcut}

As the first implication of our shifting mechanism, we demonstrate in this section that applying Algorithm \ref{alg:shifting}, with $f(x)$ representing the exact unique minimum $s$-$t$ cut on the $n$-vertex graph encoded by $x$, results in an additive-error approximation algorithm for the minimum $s$-$t$ cut with an optimal error bound of $\tilde{\Theta}(n)$. This directly recovers the same result from \cite{DBLP:conf/nips/DalirrooyfardMN23} by a significantly different and simplified analysis, while we also be able to treat this problem as a special case within our broader framework.

Formally speaking, consider a weighted graph $G=([n],E,w)$ where $w\in \mathbb{R}_{\geq 0}^{n\choose 2}$ be the vector that encodes all non-negative edge weights. For a non-edge, we consider it as an edge with weight zero. Without the loss of generality, given any distinct $s, t\in [n]$, we assume $G$ has an unique min $s$-$t$ cut, as introducing small random perturbations with sufficient precision to each edge weight, one could ensure that any two distinct cuts have different sizes. Then, we define $f: \mathbb{R}_{\geq 0}^{n\choose 2} \rightarrow 2^{[n]\backslash \{s,t\}}$ be the function that for any $w\in \mathbb{R}_{\geq 0}^{n\choose 2}$,
\begin{align}\label{eq:deff_stcut}
    f(w) = \text{the unique min } s\text{-} t \text{ cut in the graph encoded by } w.
\end{align}
The following lemma exhibits a dominating set of $f$ of size only $O(n)$ instead of $O(n^2)$.
\begin{lemma}\label{lem:dominating_stcut}
Let $S \subseteq {n\choose 2}$ such that $S = \{\{u,v\}: u\in \{s,t\}, v\in [n]\backslash \{s,t\}\}$. Then, $S$ is a dominating set of $f$ with sensitivity $2$. 
\end{lemma}
This lemma suggests that only adding noise between $s,t$ and all non-terminals is enough to preserve privacy, as edges between terminal and non-terminals form a dominating set of $f$. In this sense, the algorithm in \cite{DBLP:conf/nips/DalirrooyfardMN23} can be viewed as a special case of our shifting mechanism. By the inspiration of Lemma \ref{lem:dominating_stcut}, we directly give the following variant of Algorithm \ref{alg:shifting}:

\begin{center}
    \begin{tcolorbox}[=sharpish corners, colback=white, width=1\linewidth]\label{alg:shifting_stcut}

        \begin{center}
      \textbf{\emph{Algorithm \ref*{alg:shifting_stcut}: The Shifting mechanism on approximating min $s$-$t$ cut}}
        \end{center}
        \vspace{6pt}
    \begin{enumerate}
  \item \textbf{Input:} A weighted graph $G = ([n],E,w)$ with non-negative edge weights, two terminals $s,t\in [n]$ and a privacy budget $\epsilon > 0$.
  \item For each terminal $b\in \{s,t\}$ and each non-terminal $u$, let $$w_{\{b,u\}} = w_{\{b,u\}} + \text{Lap}(2/\epsilon) + c\log n$$ for some large enough $c$.
  \item Compute and output the exact min $s$-$t$ cut on the noisy graph.

\end{enumerate}
  \end{tcolorbox} 
  \end{center}
In Algorithm \ref{alg:shifting_stcut}, we amplify the edge weight between terminals and non-terminals with an additive factor $c\log n$ for some large enough $c$ to make sure that with high probability, all edge weights remain non-negative in the noisy graph.
\begin{theorem}
    Algorithm \ref{alg:shifting_stcut} preserves $(\epsilon,0)$-differential privacy. Furthermore, with high probability, Algorithm \ref{alg:shifting_stcut} outputs a min $s$-$t$ cut with error at most $O(n\log n)$.
\end{theorem}
\begin{proof}
    The privacy guarantee directly comes from Theorem \ref{thm:privacy_shifting} (the privacy of shifting mechanism), Lemma \ref{lem:dominating_stcut} (the dominating set of min $s$-$t$ cut problem) and Proposition \ref{prop:robust_against_transit} (the dominating set remains exist after transition). For the utility, recall that with high probability, the total number of edge weights added from adding noise and amplifying weights is $O(n\log n)$, which concludes the utility gurantee.
\end{proof}

\vspace{-0.6cm}

\subsection{The application of shifting mechanism in private multicut}\label{sec:app_multicut}

One of the generalizations of $s$-$t$ cut is \textit{Multicut}. In this problem, A graph is given with $k$ pairs of different terminals $\{(s_1, t_1), \cdots (s_k, t_k)\}$. The goal is to find a partition of vertex set of at most $2k$ parts, ensuring that each pair of terminals is separated (since there are at most $2k$ terminals to separate). The private multicut problem with $k\geq 2$ has not been studied before us. However, as a simple application, we use the shifting mechanism to give an additive noise mechanism for finding minimum multicut for any $k\in \mathbb{N}_+$. 

We first note that the inefficient exponential mechanism gives $\tilde{O}(n+k)$ additive error for multicut. This is because that to apply exponential mechanism, a naive way is to enumerate all $(n+k)^{2k}$ possible ways to separate $k$ pairs of terminals. This results in a total of $(n+k)^{O(k)}$ candidate solutions. Consequently, the additive error of $\tilde{O}(n+k)$ arises from standard analysis.

Next, we demonstrate that for constant $k$, the shifting mechanism provides a matching error bound, with the runtime depending on finding an exact multicut in an $n$-vertex graph. As a result, it directly gives efficient algorithm for $k = 2$, since the algorithm for non-private exact vertex multicut runs in polynomial time when $k = 1,2$~\cite{bousquet2011multicut}. This  generalizes the private algorithm for min $s$-$t$ cut. Formally speaking, again consider a graph $G = (V',E,w)$, where $V' = [n]\cup T$ and $T = \{s_1, \cdots, s_k, t_1,\cdots, t_k\}$ is the set of all terminals. For the same reason as in the previous section, we assume the multicut of a given graph is unique. Then, fix a set of terminals $T$, we define $f:\mathbb{R}_{\geq 0}^{V'\choose 2} \rightarrow 2^{V'}$ be the function that for any $w\in \mathbb{R}_{\geq 0}^{V'\choose 2}$, 
\begin{align}\label{eq:deff_multicut}
    f(w) = \text{the unique multicut in terms of } T\text{ in the graph encoded by } w.
\end{align}
The following lemma exhibits a dominating set of $f$ of size $O(k(n+k))$.

\begin{lemma}\label{lem:dominating_multi-cut}
Let $S \subseteq {V'\choose 2}$ such that $S = \{\{u,v\}: u\in \{s_i,t_i\} (\forall i\in [k]), v\in V'\backslash \{s_i,t_i\}\}$. Then, $S$ is a dominating set of $f$ with sensitivity $2$. 
\end{lemma}
By Theorem \ref{thm:privacy_shifting}, Lemma \ref{lem:dominating_multi-cut} directly implies an additive noise and polynomial-space pure-DP mechanism with $O(k(n+k))$ constant-variance noises. Consequently, when $k$ is constant, the additive error of this mechanism aligns with that of the exponential mechanism. Since we could use the same trick to amplify the edge weights as in the previous section to make sure that the noisy graph has all non-negative edge weights, then the shifting mechanism is efficient when $k\leq 2$ (\cite{bousquet2011multicut}). Thus, we directly give the following corollary:

\begin{corollary}
    The shifting mechanism gives an $(\epsilon,0)$-differentially private algorithm for approximating multicut with additive error $\tilde{O}(k(n+k))$. Further, the algorithm is efficient for $k = 1,2$.
\end{corollary}
\subsection{The application of shifting mechanism in private max-cut}\label{sec:app_maxcut}
In this section, we show that our shifting mechanism gives an additive noise mechanism on privately computing the max-cut on a weighted graph, with additive error matching that of the exponential mechanism. However, compared to the exponential mechanism, we are able to utilize several faster exponential-time algorithms on exact max-cut. Formally speaking, given a weighted and undirected graph $G = ([n],E,w)$, we define 
\begin{align}\label{eq:deff_maxcut}
    f(w) = \text{the unique max-cut in the graph } G\text{ encoded by } w.
\end{align}
Given another two terminals $s,t\notin [n]$, we consider a variant on max $s$-$t$ cut of graph $\hat{G} = (\hat{V},\hat{E},\hat{w})$ where $\hat{V} = [n] \cup \{s,t\}$ and define
\begin{align}\label{eq:deff_maxstcut}
    \hat f(\hat{w}) = \text{the unique max-cut of } \hat{G} \text{ in terms of } s \text{ and }t .
\end{align}
That is, for any $\hat{w}$, $\hat{f}(\hat{w})$ computes the maximum cut of the graph represented by $\hat{w}$, subject to the constraint that vertices $s$ and $t$ are in separate components. The following proposition allows us to only consider max $s$-$t$ cut on $\hat{G}$, as it approximates the max-cut on $G$, and differential privacy is robust to post-processing. Here, for a partition of vertices $\mathcal{C}\in 2^{[n]}$, we use $|\mathcal{C}(G)|$ to denote the size of the cut specified by $\mathcal{C}$ and $G$.
\begin{proposition}\label{prop:maxstcut}
    Let $G=([n],E,w)$ be any weighted graph with non-negative edge weights, and $\hat{G}$ be the graph with two extra terminals $s,t\notin [n]$. Let $\mathcal{C}\in 2^{[n]}$ be the max-cut of $G$, and $\hat {\mathcal{C}}_{st}^{G}\in 2^{[n]}$ be the projection of the max $s$-$t$ cut of $\hat{G}$ on the vertex set of $G$. Then 
    $$|\mathcal{C}(G)| \leq |\hat {\mathcal{C}}_{st}^{G}(G)| + d_s + d_t + \min\{d_s,d_t\},$$
    where $d_s$ and $d_t$ are the weighted degrees of $s$ and $t$ in $\hat{G}$, respectively.
\end{proposition}
\begin{proof}
    Let $\hat{\mathcal{C}}$ be the global max-cut on $\hat{G}$. If $s$ and $t$ have the same sign in $\hat{\mathcal{C}}$, then flip one of them gives a $s$-$t$ cut, and the size of the cut will be decreased by at most $\min\{d_s,d_t\}$. Therefore, $|\hat {\mathcal{C}}_{st}(\hat G)| \geq |\hat{\mathcal{C}}(\hat G)| - \min\{d_s,d_t\}$, where $\hat {\mathcal{C}}_{st}$ is the max $s$-$t$ cut of $\hat{G}$. This is saying that
    $$|\hat {\mathcal{C}}_{st}^G(G)| \geq |\hat {\mathcal{C}}_{st}(\hat G)| - d_s - d_t \geq |\hat{\mathcal{C}} (\hat G)| - \min\{d_s,d_t\} - d_s - d_t.$$
    On the other hand, $|\hat{\mathcal{C}}(\hat G)| \geq |\mathcal{C}(G)|$ since $\hat{G}$ has extra edges with non-negative weights. This completes the proof.
\end{proof}
Next, we study the dominating set of $\hat{f}$ on max $s$-$t$ cut. 

\begin{lemma}\label{lem:dominating_maxstcut}
Let $\hat{V} = [n]\cup \{s,t\}$, and let $S \subseteq {\hat{V}\choose 2}$ such that $S = \{\{u,v\}: u\in \{s,t\}, v\in [n]\}$. Then, $S$ is a dominating set of $\hat f$ with sensitivity $2$. 
\end{lemma}
According to Theorem \ref{thm:privacy_shifting}, Proposition \ref{prop:maxstcut} and the post-processing property of differential privacy, this lemma directly yields an additive noise mechanism with error $O(n\log n)$. We note that according to Proposition \ref{prop:robust_against_transit}, we could always assume that with high probability, the noise is positive. 

Since there are total number of $O(2^n)$ candidates for max-cut, then by standard analysis, the exponential mechanism gives $O(n\log n)$ additive error, which is matched by the shifting mechanism. Further, the run-time of the shifting mechanism depends on finding the exact max-cut on a $n$-vertex graph. Therefore, we are able to use a various of results on faster algorithms of finding max-cut in polynomial space, instead of exponential space for the exponential mechanism. In particular, \cite{scott2007linear} gives the a polynomial space algorithm for finding exact max-cut on graphs with bounded average degree $\Delta$, runs in time $O(2^{n(1-\frac{2}{\Delta+1})})$. Combining with Lemma \ref{lem:dominating_maxstcut}, and the fact that adding positive noise increases $\Delta$ by at most $2$, we have the following corollary:
\begin{corollary}
    Given any $n$-vertex graph $G$ with maximum degree $\Delta$, there is a $(\epsilon,0)$-differentially private algorithm on finding the max-cut on $G$ such that with high probability, it incurs additive error $O(n\log n)$. Further, it runs in polynomial space and time $O(2^{(n+2)(1-\frac{2}{\Delta+3})})$.
\end{corollary}

\section{Private minimum $k$-cut} \label{sec:kcut}

In this section, we study the private minimum $k$-cut problem on unweighted graph $G = ([n],E)$. In section \ref{sec:optimal_alg}, we first give a private algorithm that outputs a solution for minimum $k$-cut with purely additive error $O(k\log n/\epsilon)$. To complement this result, in Section \ref{sec:lower_bound_k_cut} we prove a nearly matching $\Omega\left(\frac{k\log \left({n}/{\log n}\right)}{\epsilon}\right)$ lower bound for $k\lesssim n^{1/2}$, which confirms that our algorithm is near-optimal. However, the algorithm does not have an efficient implementation for $k = \omega(1)$. By introducing a multiplicative error and a larger additive error, we give a polynomial time algorithm for the private minimum $k$-cut problem in Section \ref{sec:polynomial_k_cut}.

\subsection{Technical overview} \label{sec:overview-kcut}
In this section, we provide an overview of our proofs for the results on the private minimum $k$-cut problem. Our main technical contribution in this part is an $\Omega(k\log (n)/\epsilon)$ lower bound on the additive error of privately approximating a solution for minimum $k$-cut. When $k=2$, it recovers the lower bound in \cite{gupta2010differentially}. We also give a matching upper bound, which closely follows the minimum cut algorithm of Gupta et al. \cite{gupta2010differentially}.

\subsubsection{The lower bound on private minimum $k$-cut} In \cite{gupta2010differentially}, a regular graph is used to construct the lower bound for the private minimum cut problem. Specifically, they consider a $\frac{\log n }{3\epsilon}$-regular graph, such that each cut except singleton cuts has size of at least $\frac{\log n }{2\epsilon}$. Then there exists at least one vertex $u$ such that, if we remove all edges incident to it, any $\epsilon$-differentially private algorithm for the minimum cut problem will, with high probability, not select the singleton cut on $u$. Since all other cuts have size at least $\frac{\log n }{6\epsilon}$, it implies a $O(\log n/\epsilon)$ lower bound. However, this argument does not apply to the private minimum  $k$-cut problem. Since creating a $k$-cut involves removing more edges, the second smallest cut can be very small in regular graphs, making it difficult to establish a meaningful lower bound. Therefore, we consider a new graph construction to establish the lower bound: (i) first, we divide $n$ vertices into approximately $O(n/\log n)$ cliques, each of size $O(\log n)$; and (ii) second, we connect these cliques to form a ``path'', with each bridge between two cliques consisting of $d = O(\log n/\epsilon)$ edges.

By analyzing the execution of the 2-approximation algorithm in \cite{saran1995finding} on this graph (call it $G_0$), we conclude that the minimum  $k$-cut on $G_0$ is at least $d(k-1)/2$. Next, let $G_1$ and $G_2$ be two graphs with $k$ connected components, formed by removing $k-1$ bridges from $G_0$. Denote the sets of edges removed from $G_0$ to form $G_1$ and $G_2$ as $\overline{E}_1$ and $\overline{E}_2$ respectively. Now, let $\mathcal{C}$ be any $k$-partition of $n$ vertices such that the cut of $G_1$ specified by $\mathcal{C}$ has size less than $\frac{d(k-1)}{6}$. Since the cut specified by $\mathcal{C}$ on $G_0$, denoted by $\mathcal{C}(G_0)$, has size at least $\frac{d(k-1)}{2}$, then $\mathcal{C}(G_0)$ must intersect with $\overline{E}_1$ on at least $\frac{d(k-1)}{3}$ edges. Suppose $|\overline{E}_1 \cap \overline{E}_2| \leq \frac{(k-1)d}{6}$. Then, it can be shown that at least $\frac{(k-1)d}{6}$ edges will be cut off by $\mathcal{C}$ on $G_2$, and thus the size of $\mathcal{C}(G_2)$ is at least $(k-1)d/6$. This is saying that if an algorithm selects a small cut on $G_1$, this cut must cause large error on $G_2$. Due to the arbitrariness in selecting the bridges to remove, we can construct an exponential number of graphs that satisfy the condition $|\overline{E}_i \cap \overline{E}_j| \leq \frac{(k-1)d}{6}$. These ingredients enable us to use the packing argument to give an $\Omega(dk) = \Omega(k\log n/\epsilon)$ lower bound. 

\subsubsection{The optimal algorithm and the efficient algorithm on private minimum  $k$-cut} 

Our algorithm for private minimum  $k$-cut extends the ideas of the private min-cut algorithm of Gupta et al. \cite{gupta2010differentially} to $k$-cut. The privacy guarantee of \cite{gupta2010differentially} relies on a theorem of Karger, which bounds the number of $\alpha$-approximate minimum cuts by $n^{2\alpha}$. They deploy the exponential mechanism across all cuts, and argue that the cut selected is unlikely to be much larger than $\mathsf{OPT}$ by bounding the number of cuts within additive $t$ of $\mathsf{OPT}$, so long as the minimum cut is not too small. In the case that it is, they first augment the graph in a privacy-preserving way. 

In the $k$-cut case, we employ the same ideas, but instead select among all possible $k$-cuts. We first assume a \textit{black-box} bound for the number of approximate minimum  $k$-cuts. Under this assumption, we give the general form for a private $k$-cut algorithm, and prove its additive error guarantee relative to the black-box bound. We then apply a known bound of \cite{chekuri18} on the number of approximate $k$-cuts, yielding an $\epsilon$-differentially private algorithm for $k$-cut with $O(k\log{n}/\epsilon)$ error, detailed in Theorem~\ref{thm:result_kcut_upper}.

We also give a more efficient $(\epsilon, \delta)$-private version of the near-optimal algorithm that achieves the same error guarantees. To avoid enumerating the exponential number of $k$-cuts in the exponential mechanism, we draw on the ideas of \cite{gupta2010differentially} and only select from the set of sufficiently small $k$-cuts while maintaining $(\epsilon, \delta)$-privacy. These can be generated via a modified contraction algorithm of \cite{gupta21} in time $\tilde{O}(n^{k})$.

For $k = \omega(1)$, such an algorithm remains inefficient. To address this, we propose an efficient algorithm for any $k\leq n$, based on the 2-approximation algorithm SPLIT from \cite{saran1995finding}. In our algorithm, we sequentially run the private algorithm for approximately finding the minimum cut $k$ times. More formally, we start with the input graph $G = ([n],E)$. If there are already $i$ connected components of $[n]$ where $1\leq i\leq k-1$, we first use the exponential mechanism to find the connected component with the lightest minimum cut. We then privately output the minimum cut on this subgraph and split it into two new connected components. This process continues until the graph is divided into $k$ components. We will demonstrate that this approach yields a polynomial-time algorithm, albeit introducing a multiplicative error of two and a slightly larger, yet non-trivial, additive error of $\tilde{O}(k^{1.5}/\epsilon)$.


\subsection{A $\tilde{O}(n^{k})$ time private algorithm with purely additive error}\label{sec:optimal_alg}


\subsubsection{Private $k$-cut with a Black-box Analysis}\label{sec:black-box-k-cut}
We use the following assumption for our black-box analysis, and then derive an exact error bound by applying known results on the non-private $k$-cut problem~\cite{chekuri18, gupta21}.
\begin{quoting}
        Let $C$ be the cost of the min $k$-cut of $G$. There exist functions $f, g, h$ such that there are at most $n^{\alpha \cdot f(k)} \cdot g(k)^{\alpha \cdot h(k)}$ many $k$-cuts in $G$ with cost at most $\alpha C$. 
\end{quoting}


Here, we describe the general idea behind our algorithm. First, given a graph $G$, we add edges to the graph to raise the cost of the minimum $k$-cut to at least $C' = 2 \cdot \frac{f(k)\ln{n} + h(k)\ln{g(k)}}{\epsilon}$ in a differentially private manner. Second, we deploy the exponential mechanism over all $k$-cuts in the graph, invoking our black-box bound on the number of $k$-cuts within additive $t$ of $\mathsf{OPT}_k(G)$ to prove that the selected $k$-cut is not too much larger than $\mathsf{OPT}_k(G)$, where $\mathsf{OPT}_k(G)$ is the size of minimum $k$-cut on $G$. We give the following algorithm, which complements the private minimum cut algorithm in \cite{gupta2010differentially}. Let $\mathcal{C}$ be any $k$-partition on $n$ vertices. We denote $\mathcal{C}(G)$ as the set of edges of $G$ that must be removed to form the connected components specified by $\mathcal{C}$. \\

\begin{center}
    \begin{tcolorbox}[=sharpish corners, colback=white, width=1.03\linewidth]\label{alg:k-cut-general}

        \begin{center}
      \textbf{\emph{Algorithm \ref*{alg:k-cut-general}: Black-box Private $k$-cut Algorithm}}
        \end{center}
        \vspace{6pt}
    \begin{enumerate}
  \item \textbf{Input}: Graph $G = (V,E)$, integer $k \ge 3$ and privacy budget $\epsilon > 0$.
  \item {Let} $H_0 \subset H_1, \dots, \subset H_{n \choose 2}$ be arbitrary strictly increasing sets of edges on $V$.
  \item {Choose} index $i \in [0, {n \choose 2}]$ with probability $\propto \exp\left(-\epsilon \left| \mathsf{OPT}_k(G \cup H_i) - 4\frac{f(k)\ln{n} + h(k)\ln{g(k)}}{\epsilon} \right|\right)$.
  \item {Choose} a $k$-cut $\mathcal{C}$ with probability $\propto \exp(-\epsilon |\mathcal{C}(G \cup H_i)|)$. 
  \item \textbf{Return} $\mathcal{C}$.
\end{enumerate}
  \end{tcolorbox} 
  \end{center}



\begin{theorem}
    \label{theorem:privacy-general}
    Algorithm \ref{alg:k-cut-general} is $2\epsilon$-differentially private.
\end{theorem}

\begin{proof}
    Clear due to basic composition (Lemma \ref{lem:adaptive_composition}) and the privacy guarantee of the exponential mechanism (Lemma \ref{lem:exp}).
\end{proof}

For the utility, the following two lemmas assert that, with high probability, augmenting the input graph does not significantly increase size of the minimum $k$-cut, and subsequently, the selected $k$-cut gives a good result.
\begin{lemma}
    \label{lemma:augmentation-general}
    With probability at least $1 - n^{-f(k)+2}g(k)^{-h(k)}$, we have:
    \begin{equation*}
        \begin{aligned}
            2 \cdot \frac{f(k)\ln{n} + h(k)\ln{g(k)}}{\epsilon} < \mathsf{OPT}_k(G \cup H_i) < \mathsf{OPT}_k(G) + 2 \cdot \frac{f(k)\ln{n} + h(k) \ln{g(k)}}{\epsilon}.
        \end{aligned}
    \end{equation*}
\end{lemma}

\begin{proof}
    First suppose $\mathsf{OPT}_k(G) < 4 \cdot \frac{f(k)\ln{n} + h(k) \ln{g(k)}}{\epsilon}$. Then there exists some $i$ such that $$\mathsf{OPT}_k(G \cup H_i) - 4 \cdot \frac{f(k)\ln{n} + h(k) \ln{g(k)}}{\epsilon} = 0.$$ 
    The utility theorem for the exponential mechanism states that for any range $R$, output $r$, and sensitivity-1 cost score $q$, we have:
   \begin{equation*}\label{eq:exp_utility}
    \begin{aligned}
        \prob{s(D,\mathcal{E}(D)) \ge OPT + \frac{2}{\epsilon}(\ln{R} + t)} \le \exp(-t)
    \end{aligned}
   \end{equation*}
   Here, $\log{R} \le 2 \log {n}$. Let $s(k) = f(k)\log{n} + h(k) \log(g(k))$. Setting $t = (f(k)-2)\log{n} + h(k)\log(g(k))$, we have:
    \begin{align*}
        \prob{\left|OPT(G \cup H_i) - 4 \cdot \frac{s(k)}{\epsilon}\right| \ge 2\cdot\frac{s(k)}{\epsilon}} \le \frac{1}{n^{f(k)-2}g(k)^{h(k)}},
        \end{align*}
    which implies that
    $$\prob{2\cdot \frac{s(k)}{\epsilon} < OPT(G \cup H_i) < 6\cdot \frac{s(k)}{\epsilon}} \ge 1 - \frac{1}{n^{f(k)-2}g(k)^{h(k)}}.$$
        Now suppose $\mathsf{OPT}_k(G) \ge 4 \cdot \frac{f(k)\ln{n} + h(k) \ln{g(k)}}{\epsilon}$. We know that $\mathsf{OPT}_k(G \cup H_0) = \mathsf{OPT}_k(G)$. Applying the utility theorem of the exponential mechanism (\cref{eq:exp_utility}) again, we have:
        \begin{align*}
        &\prob{\left|\mathsf{OPT}_k(G \cup H_i) - 4 \cdot \frac{s(k)}{\epsilon}\right| \ge \mathsf{OPT}_k(G) - 2 \cdot \frac{s(k)}{\epsilon}} \le \frac{1}{n^{f(k)-2}g(k)^{h(k)}},
    \end{align*}
    which implies that
    $$\prob{- \mathsf{OPT}_k(G) + 6\cdot \frac{s(k)}{\epsilon} < \mathsf{OPT}_k(G \cup H_i) < \mathsf{OPT}_k(G) + 2\cdot \frac{s(k)}{\epsilon}} \ge 1 - \frac{1}{n^{f(k)-2}g(k)^{h(k)}} .$$
    Combining both cases completes the proof of Lemma \ref{lemma:augmentation-general}.
\end{proof} 

\begin{lemma}\label{lem:analyze_exp}
    If $\mathsf{OPT}_k(G \cup H_i) \ge 2 \cdot \frac{f(k)\ln{n}+h(k)\ln{g(k)}}{\epsilon}$ and $\epsilon < 1$ then:
    \begin{align*}
        \prob{|\mathcal{C}(G \cup H_i)|\ge \mathsf{OPT}_k(G \cup H_i) + b} \le \frac{1}{n^2}
    \end{align*}
    for $b = O(\frac{f(k)\ln{n}+h(k)\ln{g(k)}}{\epsilon})$.  
\end{lemma} 

\begin{proof}
    Let $c_t$ denote the number of $k$-cuts of size OPT$(G \cup H_i) + t$. Given that $\mathsf{OPT}_k(G \cup H_i) \ge C'$, by our assumption, there are at most $n^{(1+\frac{t}{C'}) \cdot f(k)} \cdot g(k)^{(1+\frac{t}{C'}) \cdot h(k)}$ such cuts. By definition, we are guaranteed that a cut of size $\mathsf{OPT}_k(G \cup H_i)$ exists; thus each cut of cost $\mathsf{OPT}_k(G \cup H_i)+t$ will receive cost at most $\exp(-\epsilon t)$ in the exponential mechanism. Therefore we have:
    \begin{align*}
        \prob{|\mathcal{C}(G \cup H_i)| \ge \mathsf{OPT}_k(G \cup H_i) + b} &\le \sum_{t \ge b} \exp(-\epsilon t) \cdot (c_t - c_{t-1}) \\
        &= \left( \sum_{t \ge b} c_t (\exp(-\epsilon t) - \exp(-\epsilon (t+1))) \right) - \exp(\epsilon b) c_{b-1} \\
        &\le \sum_{t \ge b} c_t (\exp(-\epsilon t) - \exp(-\epsilon (t+1))) \\
        &= (1-\exp(-\epsilon)) \sum_{t \ge b} c_t \exp(-\epsilon t) \\
        &\le (1-\exp(-\epsilon)) \sum_{t \ge b} n^{(1+\frac{t}{C'}) \cdot f(k)} \cdot g(k)^{(1+\frac{t}{C'}) \cdot h(k)} \exp(-\epsilon t) \\
        &= (1-\exp(-\epsilon)) n^{f(k)} g(k)^{h(k)} \sum_{t \ge b} \exp \left( \ln{\left(n^{\frac{tf(k)}{C'}} g(k)^{\frac{th(k)}{C'}}\right)} -\epsilon t \right) 
        \intertext{Now using $C' = 2 \cdot \frac{f(k)\ln{n} + h(k)\ln{g(k)}}{\epsilon}$,}
        &= (1-\exp(-\epsilon)) n^{f(k)} g(k)^{h(k)} \sum_{t \ge b} \left(\exp \left(-\epsilon/2\right) \right)^t \\
        &= (1-\exp(-\epsilon)) n^{f(k)} g(k)^{h(k)} \frac{(\exp(-\epsilon/2))^b}{1-\exp(-\epsilon/2)}
        \intertext{Let $b = 2 \cdot \frac{(f(k) + 2)\ln{n}+h(k)\ln{g(k)}}{\epsilon}$. Noting that $(1-\exp(-\epsilon)) / (1-\exp(-\epsilon/2)) = \Omega(1)$, we have:}
        &= n^{f(k)} \cdot g(k)^{h(k)} \cdot \exp(-(f(k) + 2)\ln{n}-h(k)\ln{g(k)})) \\
        &= n^{f(k)} \cdot g(k)^{h(k)} \cdot n^{-f(k) - 2} g(k)^{-h(k)}\\ 
        &= \frac{1}{n^2}.
    \end{align*}
\end{proof} 
\noindent Based on Lemma \ref{lemma:augmentation-general} and Lemma \ref{lem:analyze_exp}, we give the utility guarantee of our black-box algorithm.
\begin{theorem}
    \label{theorem:error-general}
    For any graph $G$, except with probability at most $n^{-2} + n^{-f(k)+2}g(k)^{-h(k)}$, the cost of $\mathcal{C}$ is at most $\mathsf{OPT}_k(G) + O\left(\frac{f(k)\ln{n} + h(k)\ln{g(k)}}{\epsilon}\right)$. 
\end{theorem}

\begin{proof}

Putting together the two lemmas, we have that with probability at least $1- (n^{-2} + n^{-f(k)+2}g(k)^{-h(k)})$, 
\begin{align*}
    |\mathcal{C}(G)| \le |\mathcal{C}(G \cup H_i)| &< \mathsf{OPT}_k(G \cup H_i) + 2 \cdot \frac{(f(k)+2)\ln{n}+h(k)\ln{g(k)}}{\epsilon} \\
    &< \mathsf{OPT}_k(G) + 2\cdot \frac{f(k)\ln{n}+h(k)\ln{g(k)}}{\epsilon} + 2 \cdot \frac{(f(k)+2)\ln{n}+h(k)\ln{g(k)}}{\epsilon} \\
    &= \mathsf{OPT}_k(G) + 4\cdot \frac{f(k)\ln{n}+h(k)\ln{g(k)}}{\epsilon} + \frac{4\ln{n}}{\epsilon} \\
    &= \mathsf{OPT}_k(G) + O\left(\frac{f(k)\ln{n}+h(k)\ln{g(k)}}{\epsilon}\right).
\end{align*}
\end{proof}

\subsubsection{Algorithm for private minimum $k$-cut}

Here, we combine a result of Chekuri et al. on the number of approximate $k$-cuts with the black-box analysis to prove an exact error bound. Improving on the result of \cite{DBLP:conf/soda/Karger93}, Chekuri et al. prove the following:

\begin{lemma}
    \label{theorem:karger-k-cuts}
    \text{(\cite{chekuri18})}. The number of $k$-cuts within a factor of $\alpha$ of the optimum is $O(n^{\floor{2\alpha(k-1)}})$.
\end{lemma}

We combine Chekuri et al.'s bound with the black-box analysis, where $f(k) = 2(k-1)$ and $ h(k), g(k) = 0$. This yields the following algorithm: \\

\begin{center}
    \begin{tcolorbox}[=sharpish corners, colback=white, width=1\linewidth]\label{alg:k-cut-chekuri}

        \begin{center}
      \textbf{\emph{Algorithm \ref*{alg:k-cut-chekuri}: Private $k$-cut Algorithm}}
        \end{center}
        \vspace{6pt}
    \begin{enumerate}
  \item \textbf{Input}: $G = (V,E)$, integer $k \ge 3$ and privacy budget $\epsilon > 0$.
  \item {Let} $H_0 \subset H_1, \dots, \subset H_{n \choose 2}$ be arbitrary strictly increasing sets of edges on $V$.
  \item {Choose} index $i \in [0, {n \choose 2}]$ with probability $\propto \exp\left(-\epsilon \left| \mathsf{OPT}_k(G \cup H_i) - \frac{8(k-1)\ln{n}}{\epsilon}\right|\right)$.
  \item {Choose} a $k$-cut $\mathcal{C}$ with probability $\propto \exp(-\epsilon |\mathcal{C}(G \cup H_i)|)$. 
  \item \textbf{Return} $\mathcal{C}$.
\end{enumerate}
  \end{tcolorbox} 
  \end{center}

\begin{theorem}
    For any graph $G$, with high probability, Algorithm \ref{alg:k-cut-chekuri} outputs a $k$-cut of cost at most $\mathsf{OPT}_k(G) + O(k\ln{n} / \epsilon)$.
\end{theorem}

\begin{proof}
    Combining Lemma \ref{theorem:karger-k-cuts} with Theorem \ref{theorem:error-general} gives the following bound on the additive error, with probability at least $1 - O(1/n^2)$ for any $k\geq 3$: 
    \begin{align*}
   |\mathcal{C}(G)| \le |\mathcal{C}(G \cup H_i)| &< \mathsf{OPT}_k(G) + 4\cdot \frac{2(k-1)\ln{n}}{\epsilon} +  \frac{4\ln{n}}{\epsilon} \\
    &= \mathsf{OPT}_k(G) + \frac{8k\ln{n}}{\epsilon} - \frac{4\ln{n}}{\epsilon} \\
    &= \mathsf{OPT}_k(G) + O\left(\frac{k\ln{n}}{\epsilon}\right). 
\end{align*}
\end{proof}

\begin{remark}
    \cite{gupta21} improve upon Chekuri et al. with a bound of $n^{\alpha k} k^{O(\alpha k^2)}$ on the number of approximate $k$-cuts, removing the factor of 2 in the exponent, which would yield $4\cdot \frac{(k+1)\ln{n} + O(k^2)\ln{k}}{\epsilon}$ additive error. However, the use of Chekuri et al.'s bound is justified by the fact that it gives error that is only worse by a constant for small $k$, and is better for large $k$ (i.e., $k \ge \sqrt{n}$). 
\end{remark}

We now consider the efficiency of the algorithm. We run the exponential mechanism twice. The first run selects among a universe of $n \choose 2$ elements and runs efficiently; the second selects among a universe of $k$-cuts, the number of which is given by the Stirling partition number—approximately $k^n$—and does not run efficiently. We follow the techniques of \cite{gupta2010differentially} for their efficient minimum cut algorithm, describing how to achieve $(\epsilon, \delta)$-differential privacy in $n^k\polylog{n}$ time, using the following theorem:

\begin{theorem}
    \label{theorem:gupta-runtime}
    \text{\normalfont(\cite{gupta21})}. Let $C$ be the cost of the minimum $k$-cut. For each $k \ge 3$, there is an algorithm to enumerate all $k$-cuts of weight at most $\alpha C$ in time $n^{\alpha k}(\log{n})^{O(\alpha k^2)}$ with probability at least $1 - 1 /$\text{\normalfont poly}$(n)$. 
\end{theorem}

Our more efficient algorithm is as follows: in Step 4 of Algorithm \ref{alg:k-cut-chekuri}, instead of sampling amongst all possible $k$-cuts, we restrict ourselves to the set $S_\alpha$, defined as the set of $k$-cuts generated by applying Theorem \ref{theorem:gupta-runtime} with $\alpha = 1 + \frac{k}{(k-1)}$. We claim that the output distribution of this algorithm has statistical distance $O(1 / n^2)$ from that of Algorithm \ref{alg:k-cut-chekuri}, implying that the efficient algorithm preserves $(\epsilon, O(1 / n^2))$-differential privacy. 

Consider a hypothetical algorithm that generates the $k$-cut as in Algorithm \ref{alg:k-cut-chekuri} but then outputs FAIL whenever this cut is not in $S_\alpha$. We first claim that the probability that this algorithm outputs FAIL is $O(1/n^2)$. By Lemma \ref{lemma:augmentation-general}, $\mathsf{OPT}_k(G\cup H_i)$ is at least $\frac{4(k-1)\ln{n}}{\epsilon}$ except with probability $1/n^2$. Conditioned on this, Theorem \ref{theorem:error-general} states that the $k$-cut chosen in Step 4 has cost at most $\mathsf{OPT}_k(G \cup H_i) + \frac{4k\ln{n}}{\epsilon}$ except with probability $1/n^2$. Finally, given our setting for $\alpha$, all $k$-cuts of cost at most $\mathsf{OPT}_k(G \cup H_i) + \frac{4k\ln{n}}{\epsilon}$ are enumerated in $S_\alpha$ except with probability $1 /$\text{\normalfont poly}$(n)$, proving the claim that the probability of FAIL is $O(1/n^2)$. 

We can couple this hypothetical algorithm with both the more efficient and exponential time algorithms, and observe that their output distributions are $\delta$-close for $\delta = O(1/n^2)$, implying $(\epsilon, \delta)$-privacy. The runtime of the efficient algorithm is $n^{\alpha k}(\log{n})^{O(\alpha k^2)}$ for $\alpha = 1 + \frac{k}{(k-1)} = \Omega(1)$.

\subsection{Lower bound on private minimum $k$-cut}\label{sec:lower_bound_k_cut}
In this section, we establish the ${\Omega}(k\log n)$ lower bound on privately answering minimum $k$-cut for unweighted graphs. Prior to this, we briefly cover the proof of the $\Omega(\log n)$ lower bound for private minimum cut in \cite{gupta2010differentially}, and highlight the key technical differences between their approach and ours. 

In \cite{gupta2010differentially}, the authors consider a $\frac{\log n}{3\epsilon}$-regular graph, in which the minimum cuts are the $n$ singleton cuts, and that any other cut has size at least $\frac{\log n}{2\epsilon}$. Suppose $\mathcal{M}$ is a randomized mechanism on $G = (V,E)$ that outputs a partition of vertices. Then there exists a $v\in V$ such that 
$$\Pr{\mathcal{M}(G) = (\{v\}, V\backslash \{v\})} \leq 1/n.$$
Consider a graph $G'$ formed from $G$ by removing all edges incident on $v$. Suppose $\mathcal{M}$ is $(\epsilon,0)$-differentially private, then 
$$\Pr{\mathcal{M}(G') = (\{v\}, V\backslash \{v\})} \leq 1/n^{1/3}$$
since $G$ and $G'$ differs in $\log n/{(3\epsilon)}$ edges. This is saying that with probability at least $1-\frac{1}{n^{1/3}}$, $\mathcal{M}$ will output a cut other than the minimum cut $ (\{v\}, V\backslash \{v\})$ of size zero. Since any other cut has size at least $\frac{\log n}{6\epsilon}$, then the expected error on private global min cut is at least $O(\log n/\epsilon)$.

However, this approach for constructing the hard case does not directly apply to the minimum $k$-cut problem. Say we remove edges from a similarly constructed regular graph to create $k-1$ isolated vertices, forming a $k$-cut. Even if the mechanism does not output the optimal $k$-cut, the size of the second-smallest cut might be significantly smaller than $O(k\log n)$.

Therefore, to establish a lower bound for private $k$-cut, we must consider a graph where the actual minimum $k$-cut in the original graph does not have excessively small weights, even when a specified subset of edges is removed to form $k$ connected components. We use a new construction to give the following theorem:
\begin{theorem}\label{thm:lb_on_k_cut}
  Fix any $n,k\in \mathbb{N}_+$, $3\leq k\lesssim n^{1/2}$ and $\epsilon>\frac{1}{48\ln n }$. If $\mathcal{M}$ is a randomized algorithm on $n$-vertex graphs such that on any input graph $G = ([n],E)$, it outputs a $k$-partition $\mathcal{C}$ where
  $$\mathbb{E}_{\mathcal{M}(G)} [|\mathcal{C}(G)|] \leq \mathsf{OPT}_k(G) + \frac{(k-1)\ln\left(\frac{n}{2\ln^2 n}\right)}{576\epsilon},$$ then $\mathcal{M}$ is not $\epsilon$-differentially private. Here, $|\mathcal{C}(G)|$ is the cost of cut $\mathcal{C}$ on $G$, and $\mathsf{OPT}_k(G)$ is the size of minimum $k$-cut on $G$.
\end{theorem}

\noindent Before getting into the proof of \Cref{thm:lb_on_k_cut}, we first consider the unweighted graph $G_0$ on $n$-vertices that is constructed as follows:
\begin{enumerate}
  \item Arbitrarily divide $n$ vertices into $\lfloor n/(2\ln^2 n) \rfloor$ groups.  Within each group, add edges between every pair of vertices to form a clique with size of at least $2\ln^2 n$.
  \item Sequentially add edges to connect each clique, forming a path of length $\lfloor n/(2\ln^2 n) \rfloor - 1$. Each bridge between two directly connected cliques consists of $d = {\frac{1}{48\epsilon}}\ln \left(\frac{n}{2\ln^2 n}\right)$ edges.
\end{enumerate}

Using the 2-approximation algorithm, we have the following lemma on the size of the minimum $k$-cut in such graph, which plays an important role in the proof of \Cref{thm:lb_on_k_cut}.

\begin{lemma}\label{lem:min-k-cut_in_G_0}
  Fix any $2\leq k\leq \lfloor n/(2\ln^2 n) \rfloor - 1$ and $\epsilon > \frac{1}{48\ln n}$. The size of minimum $k$-cut in $G_0$ is at least $\frac{(k-1)d}{2}$.
\end{lemma}
\begin{proof}
  Given a connected graph $G = ([n],E)$, we first introduce the following SPLIT algorithm (\cite{saran1995finding}) on approximating the minimum $k$-cut of $G$:
  \begin{itemize}
      \item In each iteration, pick the lightest cut that splits a connected component. Remove the edges in this cut.
      \item Halt when there are exactly $k$-connected components.
  \end{itemize}
  It is known in \cite{saran1995finding} that the SPLIT algorithm always gives a 2-approximation of the minimum $k$-cut on $G$. We consider executing the SPLIT algorithm on $G_0$. We claim that the algorithm will find a $k$-cut of size exactly $(k-1)d$, and we will prove this fact by induction. In the first iteration, the algorithm simply finds the minimum cut in $G_0$. For any $\epsilon > \frac{1}{48\ln n}$, cutting any bridge removes $d < \ln^2 n $ edges. Since each clique in $G_0$ has $2\ln^2 n$ vertices, then isolating a vertex requires the removal of $2\ln^2 n-1 > d$ edges. Therefore, the minimum cut of $G_0$ is any bridge with $d$ edges. In iteration $i$ where $2\leq i\leq k-1$, suppose there are $i$ connected components in which each connected component is either a clique of $2\ln^2 n$ vertices, or a path of some such cliques connected with bridges of $d$ edges. Since the minimum cut in such clique has size at least $2\ln^2 n-1$, then the algorithm will again remove all edges in a bridge to split a path.
  Therefore, executing the SPLIT algorithm on $G_0$ is equivalent to choosing arbitrary $k-1$ bridges and removing the edges of them. Thus, the total cost is $(k-1)d$. Since this greedy algorithm gives a $2$-approximation of the actual min $k$-cut, then we conclude that $\mathsf{OPT}_k(G_0)\geq \frac{(k-1)d}{2}$.
\end{proof}

Next, let $\ell = \lfloor n/(2\ln^2 n) \rfloor - 1$. We fix an arbitrary $k$ that is no larger than $ \sqrt{\ell}/2 + 1 = \tilde{\Omega}(n^{1/2})$. Given any subset $S\subseteq [\ell]$ of size $k-1$, we define $G_S$ as the graph that is formed from $G_0$ by removing the $k$ bridges specified by $S$. Clearly for any $S\subseteq [\ell]$ and $|S| = k-1$, $G_S$ has exactly $k$ connected components, and thus the size of minimum $k$-cut on $G_S$ is zero. Let $\mathcal{G}$ be a set of graphs such that $$\mathcal{G} \subseteq \{G_S||S| = k-1 \land S\subseteq [\ell]\},$$ and for any pair of $G_{S}$ and $G_{S'}$ in $\mathcal{G}$, $$|S \cap S'| \leq (k-1)/6.$$ That is, the graphs in $\mathcal{G}$ do not overlap with each other very much. Intuitively, this says that if $\mathcal{C}$ is a good $k$-cut on $G_S$, then it must cause large error on any other graph $G_{S'}$. We will formulate this later. The following lemma given by Frankl et al. establishes the existence of such a set $\mathcal{G}$ with exponentially many elements.

\begin{lemma}[\cite{frankl2016invitation}]\label{lem:Frankl}
  Fix any $\ell \geq 4$. For every $2\leq k \leq \sqrt{\ell}/2 + 1$, there exists a $(k-1)$-uniform family $\mathcal{G}$ of size at least $\left(\frac{\ell}{2(k-1)}\right)^{m}$ on $\ell$ elements such that $|S\cap S'|\leq m-1$ for any two distinct sets $S,S' \in \mathcal{G}$.
\end{lemma}
\noindent Then, by setting $m-1  = \lfloor (k-1)/6 \rfloor $, from Lemma \ref{lem:Frankl}, we directly have the following corollary:

\begin{corollary}\label{cor:size_of_G}
  There exists such a set of graphs $\mathcal{G}$ satisfying previous conditions of size $|\mathcal{G}| \ge \ell^{k/12}$.
\end{corollary}

Next, we show that a fixed $k$-cut cannot be too ``good'' simultaneously for any two graphs in $\mathcal{G}$. Let $\mathcal{C}$ be a $k$-partition on $n$ vertices. For any $G = ([n],E)$, denote $\mathcal{C}(G)$ as the edge set removed from $E$ to form $k$ connected components, and its size on $G$ as $|\mathcal{C}(G)|$.

\begin{lemma}\label{lem:contradiction}
  Let $\mathcal{C}$ be any $k$-partition on $n$ vertices. Fix any $G_S\in \mathcal{G}$ where $S \subseteq [\ell]$ and $|S| = k-1$. If $|\mathcal{C}(G_S)| < \frac{(k-1)d}{6}$, then $|\mathcal{C}(G_{S'})| > \frac{(k-1)d}{6}$ for any $G_{S'}\in \mathcal{G}$ and $S'\neq S$.
\end{lemma}
\begin{proof}
  We remark that the actual minimum $k$-cut on $G_S$ is zero. Let $\overline{E}_S$ be the set of edges removed from $G_0$ to form $G_S$, and $E_S$ be the edge set of $G_S$. By Lemma \ref{lem:min-k-cut_in_G_0}, we see that the size of minimum $k$-cut in $G_0$ satisfies $\mathsf{OPT}(G_0) \geq \frac{(k-1)\log n}{2}$; namely for any $k$-partition $\mathcal{C}$, $|\mathcal{C}(G_0)|\geq \frac{(k-1)d}{2}$. Therefore, by the assumption that $|\mathcal{C}(G_S)|< \frac{(k-1)d}{6}$, we have that 
  \begin{equation}\label{eq:4.1}
    |\mathcal{C}(G_0) \cap \overline{E}_S| = |\mathcal{C}(G_0)| - |\mathcal{C}(G_S)| > \frac{(k-1)d}{3}.
  \end{equation}
  Let $G_{S'}$ be any other graph in $\mathcal{G}$ such that $S'\neq S$. Since $|S\cap S'|\leq \frac{k-1}{6}$, then 
  \begin{equation}\label{eq:4.2}
    |\overline{E}_S \cap \overline{E}_{S'}|\leq \frac{(k-1)d}{6}
  \end{equation}
   by the construction of $G_S$ and $G_{S'}$. By the fact that $E_S \cup \overline{E}_{S} = E_{S'} \cup \overline{E}_{S'} =E$, we have
   $$\overline{E}_{S} - \overline{E}_{S'} = E_{S'} - E_S  \subseteq E_{S'}.$$
  Therefore, 
  \begin{equation*}
    \begin{aligned}
      |\mathcal{C}(G_{S'})| = |\mathcal{C}(G_0) \cap E_{S'}| &\geq |\mathcal{C}(G_0) \cap (\overline{E}_{S}\backslash \overline{E}_{S'})|\\
      &\geq |\mathcal{C}(G_0) \cap \overline{E}_S| - |\mathcal{C}(G_0)\cap (\overline{E}_S \cap \overline{E}_{S'})|\\
      &\geq |\mathcal{C}(G_0) \cap \overline{E}_S| - |\overline{E}_S \cap \overline{E}_{S'}|\\
      & > \frac{(k-1)d}{3} - \frac{(k-1)d}{6} = \frac{(k-1)d}{6},
    \end{aligned}
  \end{equation*}
  which completes the proof. Here, the last inequality comes from \cref{eq:4.1} and \cref{eq:4.2} respectively.
\end{proof}

\noindent With all of these technical ingredients, we are ready to give the proof of \Cref{thm:lb_on_k_cut}.
\begin{proof}
  (Of \Cref{thm:lb_on_k_cut}.) Let $\alpha = \frac{(k-1)d}{12}$. For any $S$ subset $[\ell]$ and $|S| = k-1$, we denote by $B_S^\alpha$ the set of $k$-partitions $\mathcal{C}$ on $V = [n]$ such that the cost of $\mathcal{C}$ on $G_S$ is less than $2\alpha$:
  $$B_S^\alpha = \{\mathcal{C}: \mathcal{C} \text{ is a }k\text{-partition and } |\mathcal{C}(G_S)|< 2\alpha\}.$$
  Since the mechanism $\mathcal{M}$ on $G$ outputs a $\mathcal{C}$ which satisfies that 
  \begin{equation*}
      \begin{aligned}
          \mathbb{E}_{\mathcal{M}}[|C(G)|] < \mathsf{OPT}_k(G) + \frac{(k-1)\ln\left(\frac{n}{2\ln^2 n}\right)}{576\epsilon}  &= \mathsf{OPT}_k(G) + \frac{(k-1)d}{12}\\
          &= \mathsf{OPT}_k(G) + \alpha,
      \end{aligned}
  \end{equation*}
  then on any $G_S$, with probability at least $\frac{1}{2}$, $\mathcal{M}(G_S)\in B_S^\alpha$ since $\mathsf{OPT}_k(G_S) = 0$.
  Equivalently, let $\mu_S(\cdot)$ be the distribution of $\mathcal{M}(G_S)$ for any $G_S\in \mathcal{G}$, then $\mu_S(B_S^\alpha) \geq 1/2$. 
  
  For the sake of contradiction, suppose the mechanism $\mathcal{M}$ is $(\epsilon,0)$-differentially private in terms of deleting or adding an edge. Then, for any possible collection $\mathcal{B}$ of $k$-partitions on $[n]$ vertices and any $G_S,G_{S'}\in \mathcal{G}$, 
  $$\mu_S(\mathcal{B}) \geq \mu_{S'}(\mathcal{B})e^{-\epsilon\cdot \frac{k-1}{24\epsilon}\ln \left(n/(2\ln^2 n)\right)} = \mu_{S'}(\mathcal{B})e^{-\frac{k-1}{24}\ln \left(n/(2\ln^2 n)\right)}.$$
  This is because transferring from $G_S$ to $G_{S'}$ requires editing $2d = \frac{k-1}{24\epsilon}\ln \left(\frac{n}{2\ln^2 n}\right)$ edges. On the other hand, by Lemma \ref{lem:contradiction}, $B_S^\alpha \cap B_{S'}^\alpha = \emptyset$ for any pair $G_S$ and $G_{S'}$ in $\mathcal{G}$. Therefore, fix any $G_S\in \mathcal{G}$, 
  \begin{equation*}
    \begin{aligned}
      \mu_S\left(\bigcup_{G_{S'}\in \mathcal{G}} B_{S'}^\alpha \right) = \sum_{G_{S'}\in \mathcal{G}}\mu_S(B_{S'}^\alpha) &\geq e^{-\frac{k-1}{24}\ln \left(n/(2\ln^2 n)\right)}\cdot \sum_{G_{S'}\in \mathcal{G}} \mu_{S'}(B_{S'}^\alpha)\\
      &\geq \frac{1}{2}  e^{-\frac{k-1}{24}\ln \left(n/(2\ln^2 n)\right)}\cdot |\mathcal{G}|\\
      & \geq \frac{1}{2} e^{-\frac{k-1}{24}\ln \left(n/(2\ln^2 n)\right)} \cdot \left(\frac{n}{2\ln^2 n}\right)^{\frac{k}{12}} >1,
    \end{aligned}
  \end{equation*}
  which leads to a contradiction. Here, the second to last inequality comes from Corollary \ref{cor:size_of_G}.
\end{proof}

\subsection{A polynomial time and private algorithm with 2-approximation}\label{sec:polynomial_k_cut}

In this section, we show that there exists an efficient $(\epsilon,\delta)$-differentially private algorithm such that on any unweighted graph $G$, it outputs a $k$-cut $\mathcal{C}$ such that 
$$|\mathcal{C}(G)|\leq 2\mathsf{OPT}_k(G) + \tilde{O}\left(\frac{k^{3/2}}{\epsilon}\right).$$
Here, $\mathsf{OPT}_k(G)$ is the size of minimum $k$-cut on $G$. Our algorithm is a private variant of the SPLIT algorithm in Saran et al.~\cite{saran1995finding}. In the SPLIT algorithm, we sequentially choose the lightest minimum cut in each connected component, until there are at least $k$ connected components. We will repeatedly use the private min-cut algorithm:
\begin{theorem}[\cite{gupta2010differentially}]\label{thm:private_min_cut}
  Fix any $0<\epsilon<1$ and $\delta<1/n^c$ for some large constant $c$. There exists a $(\epsilon, \delta)$-differentially private algorithm $\mathcal{M}$ such that on any input unweighted graph $G = ([n],E)$, it outputs a cut $(S,\overline{S})$ where $S\subseteq [n]$ in polynomial time and with probability at least $1-\frac{1}{n^2}$, $$\mathsf{size}_G(S) \leq \mathsf{OPT}(G) + \frac{20\log n}{\epsilon}.$$
\end{theorem}

\noindent For any $G=(V,E)$ and $S\subseteq V$, we denote the sub-graph of $G$ induced by $S$ as $G(S) = (S,E')$, where $E'\subseteq E$. A natural idea is to plug the mechanism $\mathcal{M}$ in \Cref{thm:private_min_cut} into the SPLIT algorithm. In particular, on the input graph $G$, we initially just run $\mathcal{M}$ to privately output a cut that approximates the min-cut in $G$. In the following iterations, supposing we have divided the graph into $i$ pieces, we use the exponential mechanism to choose a part such that the corresponding sub-graph induced by this part has the smallest minimum cut. Then, we again run $\mathcal{M}$ on this sub-graph to create two new pieces. 

Notice that in each iteration, each part is disjoint. However, every part does not necessarily correspond to a connected component. Specifically, outputting the connected components in each iteration violates privacy, since adding or deleting an edge would possibly change the number of connected components.

 Given parameters $\epsilon,\delta$ and an unweighted graph $G$, the algorithm is listed as follows:

 \begin{center}
    \begin{tcolorbox}[=sharpish corners, colback=white, width=1\linewidth]\label{alg:approx_k_cut}

        \begin{center}
      \textbf{\emph{Algorithm \ref*{alg:approx_k_cut}: Private and efficient approximation of $k$-cut}}
        \end{center}
        \vspace{6pt}
\begin{enumerate}
  \item {Let} $\epsilon_0 = \frac{\epsilon}{6\sqrt{k\ln (2/\delta)}}$ and $\delta_0 = \frac{\delta}{2k}$.
  \item {Let} $G_0$ be the input graph, and let $\mathcal{G}_0 = \{G_0\}$.
  \item {For} $i = 1, 2, \cdots, k-1$:
  \begin{enumerate}[label=\theenumi\alph*)]
  \item {For} each graph $G_\ell$ in $\mathcal{G}_{i-1}$, compute the value of actual minimum cut $\mathsf{OPT}(G_{\ell})$.
  \item {Choose} $\ell \in \{0,1, \cdots,i-1\}$ according to probability proportional to $\exp(-\epsilon_0 \cdot \mathsf{OPT}(G_{\ell}))$. Then, let $\mathcal{G}_{i-1}' = \mathcal{G}_{i-1}\backslash \{G_{\ell}\}$.
  \item {Let} $V_{\ell} \subseteq [n]$ be the support of $G_{\ell}$. Run $\mathcal{M}(G_{\ell})$ with privacy parameter $(\epsilon_0,\delta_0)$ and report $(S^*,V_{\ell}\backslash S^*)$ where $S^* \subseteq V_{\ell}$.
  \item {Let} $G_{i}\leftarrow G_{\ell}(S^*)$ and update $G_{\ell}$ by $G_{\ell}\leftarrow G_{\ell}(V_\ell \backslash S^*)$.
  \item {Update} $\mathcal{G}_{i-1}$ to $\mathcal{G}_{i} := \mathcal{G}_{i-1}' \cup \{G_{\ell}, G_i\}$.
  \end{enumerate}

  \item After $k-1$ iterations ends, output the $k$ disjoint supports of graphs in $\mathcal{G}_{k-1}$ be $T_1,\cdots, T_{k}$.
\end{enumerate}
  \end{tcolorbox} 
  \end{center}

\noindent Clearly, Algorithm \ref{alg:approx_k_cut} always terminates in polynomial time for any $k = \omega(1)$.

\subsubsection{Privacy analysis}
In this section, we present the following theorem on the privacy guarantee of Algorithm \ref{alg:approx_k_cut}:
\begin{theorem}
  Algorithm \ref{alg:approx_k_cut} preserves $(\epsilon,\delta)$-differential privacy.
\end{theorem}
\begin{proof}
  In the $i$-th iteration where $1\leq i \leq k-1$, suppose the current set of graphs are $\mathcal{G}_{i} = \{G_0,G_1,\cdots, G_{i-1}\}$. We remark that each graph in $\mathcal{G}$ is a vertex-induced sub-graph of $G$. Let $\mathcal{S}_i = \{S_0, S_1,\cdots, S_{i-1}\}$ be the corresponding sets of vertices of graphs in $\mathcal{G}_i$. It can be verified that $\mathcal{S}_i$ can be fully recovered by the outputs of the algorithm in the first $i-1$ iterations. Then, suppose $G'$ is formed by adding or deleting an edge in $G$, and let the set of graphs induced by $\mathcal{S}_i$ on $G'$ be $\mathcal{G}_i' = \{G_0', G_1', \cdots, G_{i-1}'\}$. Clearly for each $0\leq \ell\leq i-1$, $G_{\ell}$ and $G_{\ell}'$ differ in at most one edge. Therefore, in Step 3b, followed by the exponential mechanism,
  \begin{equation*}
    \begin{aligned}
      \frac{\textbf{Pr}_{\mathcal{G}_i}[{\ell} \text{ is chosen}]}{\textbf{Pr}_{\mathcal{G}'_i}[\ell  \text{ is chosen}] } &= \frac{\exp(-\epsilon_0 \cdot \mathsf{OPT}(G_{\ell}))}{\exp(-\epsilon_0 \cdot \mathsf{OPT}(G'_{\ell}))}\cdot \frac{\sum_{a = 1}^{i-1}e^{-\epsilon_0 \mathsf{OPT}(G_{a})}}{\sum_{a = 1}^{i-1}e^{-\epsilon_0 \mathsf{OPT}(G'_{a})}} \\
      &\leq \exp(\epsilon_0 + \epsilon_0 |\mathsf{OPT}(G_\ell) - \mathsf{OPT}(G_{\ell}')|)\leq e^{2\epsilon_0}.
    \end{aligned}
  \end{equation*}
  In Step 3c, the algorithm runs $\mathcal{M}$ on $G_{\ell}$ and outputs $(S^*, V\backslash S^*)$. Once again, since $G_{\ell}$ and $G'_{\ell}$ differ in at most one edge, this step preserves $(\epsilon_0,\delta_0)$-differential privacy. In each iteration, deciding an $\ell\in \{0,1,\cdots, i-1\}$ and outputting $(S^*, V_{\ell}\backslash S^*)$ is $(3\epsilon_0, \delta_0)$-differentially private by basic composition. Since there are $k-1$ iterations, {by advanced composition (Lemma \ref{lem:adv_composition}), the algorithm preserves $(\epsilon,\delta)$-differential privacy.}
\end{proof}

\subsubsection{Utility analysis}
Here we present the utility guarantee of Algorithm \ref{alg:approx_k_cut}. It is worth noting that the utility analysis does not directly come from the analysis of algorithm SPLIT~\cite{saran1995finding}. In the analysis of SPLIT, some properties that only hold for the optimal minimum cut are used. It is not immediately obvious whether these properties still hold for approximate minimum cuts chosen by the algorithm. For example, in each iteration of the original SPLIT algorithm, the number of connected components increases only by one, but this is not necessarily true for our algorithm, since the cut chosen in Step 3c might be a cut that separates the graph into multiple connected components. In particular, we prove the following theorem:

\begin{theorem}\label{thm:utility_twice_approx}
  Fix any $2\leq k\leq n$. Given any input graph $G$, let $\mathcal{C}$ be a $k$-partition on $[n]$ that is outputted by Algorithm \ref{alg:approx_k_cut}. With probability at least $1-2/n$, 
  $$|\mathcal{C}(G)| \leq 2\mathsf{OPT}_k(G) + \tilde{O}\left(\frac{k^{1.5}  \sqrt{\log(2/\delta)}}{\epsilon}\right).$$
\end{theorem}

To prove \Cref{thm:utility_twice_approx}, we first show that the $k$-cut chosen by our algorithm is relatively good for any sequence of cuts in the original graph that results in a $k$-cut. Let $G = ([n],E)$ be the input graph. 
For any $B\subseteq E$, denote by $\mathsf{COMPS}_G(B)$ the number of connected components of $G$ after deleting all edges in $B$. If no ambiguity arises, we will omit $G$ in subscript. For the sake of convenience, let $p_1,p_2,\cdots, p_{k-1}$ be the set of edges that is removed implicitly by the algorithm in order. That is, $p_i \subseteq \cup_{0\leq \ell \leq i-1} E_{\ell}$, where $E_\ell$ is the edge set of $G_\ell$ in $\mathcal{G}_{i-1}$. Let $b_1,\cdots, b_{k-1}$ be an arbitrary sequence of subsets of the edges $E$ (which may include duplicates) that satisfies the following property:
\begin{itemize}
  \item For each $1\leq i\leq k-1$, $b_i \subseteq E$ is a cut.
  \item $\mathsf{COMPS}(b_1 \cup b_2 \cup \cdots \cup b_{k-1}) \geq k$. That is, the union of cuts $b_1 \cup b_2 \cup \cdots \cup b_{k-1}$ is a $k$-cut. 
  \item For any $1\leq j \leq k-1$, if $b_j$ is unique in the sequence, then $\mathsf{COMPS}(b_1 \cup b_2 \cup \cdots \cup b_{j}) - \mathsf{COMPS}(b_1 \cup b_2 \cup \cdots \cup b_{j-1}) = 1$. We define $b_0 := \emptyset$.
  \item For any $1\leq j\leq k-1$, if $\mathsf{COMPS}(b_1 \cup b_2 \cup \cdots \cup b_{j+1}) - \mathsf{COMPS}(b_1 \cup b_2 \cup \cdots \cup b_{j}) = t$, then $b_{j+1}$ will repeat $\min\{t, k-j\}$ times in the sequence. That is, $b_{j+1} = b_{j+2} = \cdots = b_{\min \{j+t, k\}}$. 
\end{itemize}
Based on the last two properties, it can be verified that each new cut $b_i$ in the sequence must create new connected components, since the number of connected components increases by the multiplicity of $b_i$, after performing the union with $b_i$.
For any cut $b\subseteq E$, let $w(b) = |b|$ be the weight of $b$. We now prove the following lemma:
\begin{lemma}\label{lem:cut_series}
  Let $b_1,\cdots, b_{k-1} \subseteq E$ be any sequence of cuts (allowing duplicates) that satisfy the above properties. Then $w(p_i) \leq w(b_i) + \frac{20\log n}{\epsilon_0}$ for any $1\leq i\leq k-1$.
\end{lemma}
\begin{proof}
  We prove it by induction. This lemma is clearly true for $i=1$. By the guarantee of $\mathcal{M}$ in \Cref{thm:private_min_cut}, with probability at least $1-1/n^2$, the first cut $p_1$ specified by $G$ and $\mathcal{M}(G)$ satisfies
  $$w(p_1) \leq \mathsf{OPT}(G) + \frac{20\log n}{\epsilon_0} \leq b_1 + \frac{20\log n}{\epsilon_0}$$
  since $b_1$ is a cut. For any $1<i\leq k-1$, after deleting $i-1$ cuts, there are at least $i$ connected components in $G$. We then consider two situations:
  
  \textbf{Case(1)}: After deleting $p_1, p_2, \cdots, p_{i-1}$, there are exactly $i$ connected components in $G$. That is, each graph $G_0,\cdots, G_{i-1}$ in $\mathcal{G}_{i-1}$ is a connected component. By the properties of the sequence $b_1,b_2,\cdots, b_{k-1}$, we have that $\mathsf{COMPS}(b_1\cup b_2 \cdots \cup b_{i}) \geq i+1$, since each new cut increases the number of connected components by at least one. Since $p_1\cup p_2 \cdots \cup p_{i-1}$ only creates exactly $i$ connected components, then there exists $j$, where $1\leq j \leq i$, such that 
  $$b_j \not\subseteq p_1\cup p_2 \cdots \cup p_{i-1}.$$
  Since $b_j$ is a cut in $G$ and $b_j$ is not fully contained in $ p_1\cup p_2 \cdots \cup p_{i-1}$, then $b_j' = b_i \backslash (p_1\cup p_2 \cdots \cup p_{i-1})$ must be a cut in at least one of the graphs $G_0, \cdots, G_{i-1}$. Let $\ell^*$ be the index such that $\mathsf{OPT}(G_{\ell^*}) \leq \mathsf{OPT}(G_{\ell})$ for any $0\leq \ell \leq i-1$. Then, we have 
  $$ \mathsf{OPT}(G_{\ell^*}) \leq w(b_j') \leq w(b_j). $$
  Meanwhile, at the $i$-th iteration, Step 3b of Algorithm \ref{alg:approx_k_cut} chooses $a\in \{0,1,\cdots, i-1\}$ such that with probability at least $1-1/n^2$, $\mathsf{OPT}(G_a) \leq \mathsf{OPT}(G_{\ell^*}) + \frac{6\ln n}{\epsilon_0}$. With the utility guarantee of $\mathcal{M}(G_a)$ and the union bound, we have with probability at least $1- 2/n^2$, 
  $$|p_i| \leq \mathsf{OPT}(G_a) + \frac{20\log n}{\epsilon} \leq \mathsf{OPT}(G_{\ell^*}) + \frac{26\log n}{\epsilon_0} \leq w(b_j) + \frac{26\log n}{\epsilon_0}.$$

  \textbf{Case(2)}: After deleting $p_1, p_2, \cdots, p_{i-1}$, there are more than $i$ connected components. This means that there exists a graph in $\mathcal{G}_{i-1}$ that is not connected. Let this graph be $G_{\ell^*}$; clearly $\mathsf{OPT}(G_{\ell^*}) = 0$. Using similar reasoning as in case(1), with probability at least $1-2/n^2$,
  $$|p_i| \leq \mathsf{OPT}(G_{\ell^*}) + \frac{26\log n}{\epsilon_0} \leq w(b_i) + \frac{26\log n}{\epsilon_0}.$$
  This completes the proof of Lemma \ref{lem:cut_series}.
\end{proof}

Let $A\subseteq E$ be the minimum $k$-cut. Next, we apply the construction from Saran et al.~\cite{saran1995finding} to demonstrate the existence of a sequence of cuts $b_1, b_2, \cdots, b_{k-1}$ that satisfy all required properties and closely approximate the actual minimum $k$-cut $A$. Notice that if $G$ has more than $k$ connected components, then we just let $b_1,b_2,\cdots, b_{k-1}$ be $\emptyset$, which yields the desired result. Therefore, we suppose $G = ([n],E)$ has no more than $k$ connected components. Let $V_1, V_2,\cdots V_k$ be the $k$ connected components in $G = ([n], E-A)$. For any $1\leq i\leq k$, let $a_i$ be the cut the separates $V-V_i$ and $V_i$. We note that $a_i$ might be empty since $G$ is not necessarily connected. Without loss of generality, we assume $a_1, a_2, \cdots, a_k$ are listed in non-decreasing order. Clearly $A = \cup_{1\leq i\leq k} a_i$, and $$\sum_{i=1}^n w(a_i) = 2w(A) = 2\mathsf{OPT}_k(G),$$ since each edge in $A$ is counted twice in the sum. The following algorithm gives a construction of the desired sequence $b_1,b_2,\cdots, b_{k-1}$:
\begin{itemize}
  \item For each $e\in E$, let $s_e \in E$ be the minimum cut that separates the endpoints of $e$. Sort these cuts increasingly to get $s_1, s_2, \cdots, s_{m}$.
  \item Greedily select cuts from the list above until their union is a $k$-cut. Include a selected cut $s_i$ only if $s_i\not\subseteq s_1\cup s_2\cup \cdots, s_{i-1}$.
\end{itemize}

Let the cuts chosen by the above procedure be $b'_1,b'_2,\cdots, b'_q$. According to the rule for selecting cuts, each chosen cut $b'_i$ will generate at least one new connected component. This is because $b'_i$ is the minimum cut that separates two vertices, and thus there exists an edge $\{u,v\} \in b'_i - (b'_1\cup b'_2\cdots\cup b'_{i-1} )$ that will become disconnected from the rest of the graph after removing $b_i$. Therefore, $q\leq k-1$. 

Then, we let the sequence $b_1, b_2, \cdots b_{k-1}$ be generated from $b'_1,b'_2,\cdots, b'_q$ be repeating each $b_{i}'$ by $t$ times if it creates $t$ connected components, and truncating the sequence when its length exceeds $k-1$. It is clear that such a sequence $b_1, b_2, \cdots b_{k-1}$ satisfies all the properties listed above. We apply the following guarantee about such sequence:

\begin{lemma}[Saran et al.~\cite{saran1995finding}]
  Let $b_1, b_2, \cdots, b_k$ be the sequence of cuts constructed as above. Then
  $$\sum_{i=1}^{k-1} b_i \leq \sum_{i=1}^{k-1} a_i.$$
\end{lemma}

\noindent Putting everything all together, we have that with probability at least $1-2/n$, the algorithm outputs a $k$-partition $\mathcal{C}$ of the vertices of the input graph $G$ such that

\begin{equation*}
  \begin{aligned}
    |\mathcal{C}(G)| &= w(p_1\cup p_2 \cdots p_{k-1}) \leq \sum_{i=1}^k p_{i} \leq \sum_{i=1}^{k-1} \left( b_i + \frac{26\log n}{\epsilon_0}\right)\\
    & \leq \sum_{i=1}^k \left( a_i +  \frac{26\log n}{\epsilon_0} \right) = 2 \mathsf{OPT}_k(G) + 156 \cdot \frac{ k^{3/2} \log n \sqrt{\log(2/\delta)}}{\epsilon},
  \end{aligned}
\end{equation*}
which completes the proof of \Cref{thm:utility_twice_approx}.

\subsubsection*{Acknowledgement} We thank Yan Zhong for her discussion during the initial phase of this project. We also thank Adam Smith, Jingcheng Liu and Jalaj Upadhyay for their helpful comments and feedback on different parts of the paper draft.
\clearpage
\bibliographystyle{halpha} 
\bibliography{privacy,cut,referene}

\clearpage
\appendix

\section{Additional Related Work} \label{sec:related}
The cut problems that we study have been studied extensively in the non-private setting.  Minimum cut refers to the smallest set of edges that, if removed, would disconnect a graph into two disjoint subsets. One of the most popular algorithms for finding the minimum cut in graphs is Karger’s Algorithm~\cite{DBLP:conf/soda/Karger93}
and its subsequent refinements~\cite{DBLP:journals/jacm/KargerS96, DBLP:conf/stoc/BenczurK96, DBLP:journals/jacm/Karger00}, which have set benchmarks in this domain.   For state of the art deterministic algorithms, please refer to the near-linear time algorithms described in~\cite{DBLP:journals/jacm/KawarabayashiT19, DBLP:journals/siamcomp/HenzingerRW20} for simple graphs, and in~\cite{DBLP:conf/stoc/Li21, DBLP:conf/soda/HenzingerLRW24} for weighted graphs.

Minimum $k$-cut generalizes the concept of minimum cut, seeking the smallest set of edges whose removal partitions the graph into $k$ disjoint subsets rather than just two. This problem is particularly relevant in clustering, parallel processing, and community detection in social networks, where we aim to identify multiple distinct groups or clusters. Karger and Stein demonstrated that their randomized contraction algorithm finds a minimum $k$-cut in $\tilde{O}(n^{2k-2})$ time~\cite{DBLP:journals/jacm/KargerS96}, and their analysis was surprisingly improved to $\tilde{O}(n^{k})$ time~\cite{DBLP:journals/jacm/GuptaHLL22}, which also works for weighted graphs. For simple graphs, this $n^{k}$ barrier was overcome by~\cite{DBLP:conf/stoc/He022}.

The $s$-$t$ cut problem, also known as the minimum $s$-$t$ cut problem, is a specific case of the minimum cut problem that aims to find the smallest set of edges that, when removed, disconnects two designated nodes, $s$ (source) and $t$ (sink). Minimum $s$-$t$ is polynomially solvable using max-flow algorithms~\cite{DBLP:conf/focs/ChenKLPGS22}.
Multiway cut extends the concept further by dealing with multiple terminal nodes. The objective is to disconnect a given set of $k$ terminal nodes by removing the smallest possible set of edges, ensuring that no two terminals remain connected. Multiway cut was shown to be NP-hard when $k\geq 3$~\cite{DBLP:journals/siamcomp/DahlhausJPSY94}. 
Approximation algorithms for the multiway cut problem include a 1.5-approximation ~\cite{DBLP:journals/jcss/CalinescuKR00} based on LP relaxation and rounding methods, a $(2-2/k)$-approximation ~\cite{DBLP:journals/mp/ZhaoNI05} using a greedy splitting strategy, and a 1.2965-approximation ~\cite{sharma2014multiway} employing a more sophisticated rounding technique for LP relaxation. The current lower bound for approximating the multiway cut problem is 1.20016~\cite{DBLP:journals/mp/BercziCKM20}.


\section{Differential Privacy Basics} \label{app:DP}
Here, we give a brief exposition of differential privacy to the level required to understand the algorithms and their analysis. Differential privacy, proposed by \cite{dwork2006calibrating}, is a widely accepted standard of data privacy, which we formally define below. Let $\mathcal{D}$ be some domain of datasets. The definition of differential privacy has been mentioned in Section \ref{sec:intro}. A key feature of differential privacy algorithms is that they preserve privacy under post-processing—that is, without any auxiliary information about the dataset, any analyst cannot compute a function that makes the output less private. 
\begin{lemma}[Post processing~\cite{dwork2014algorithmic}]\label{l.post_processing}
  Let $\mathcal{A}:\mathcal{D}\rightarrow \mathcal{R}$ be a $(\epsilon,\delta)$-differentially private algorithm. Let $f:\mathcal{R}\rightarrow \mathcal{R}'$ be any function, then $f\circ \mathcal{A}$ is also $(\epsilon,\delta)$-differentially private.
\end{lemma}

Sometimes we need to repeatedly use differentially private mechanisms on the same dataset, and obtain a series of outputs.

\begin{lemma}[Adaptive composition~\cite{dwork2006calibrating}]\label{lem:adaptive_composition}
  Suppose $\mathcal{M}_1(x):\mathcal{D} \rightarrow \mathcal{R}_1$ is $(\epsilon_1,\delta_1)$-differentially private and $\mathcal{M}_2(x,y):\mathcal{D} \times \mathcal{R}_1\rightarrow \mathcal{R}_2$ is $(\epsilon_2,\delta_2)$-differentially private with respect to $x$ for any fixed $y\in \mathcal{R}_1$, then the composition
  $x \Rightarrow (\mathcal{M}_1(x), \mathcal{M}_2(x,\mathcal{M}_1(x)))$
  is $(\epsilon_1 + \epsilon_2, \delta_1 + \delta_2)$-differentially private.
\end{lemma}

\begin{lemma}
    [Advanced composition lemma, \cite{dwork2010boosting}] 
    \label{lem:adv_composition}
    For parameters $\epsilon>0$ and $\delta,\delta'\in [0,1]$, the composition of $k$ $(\epsilon,\delta)$ differentially private algorithms is a $(\epsilon', k\delta+\delta')$ differentially private algorithm, where 
    $\epsilon' = \sqrt{2k\log(1/\delta')} \cdot \epsilon + k\epsilon (e^\epsilon - 1).$
\end{lemma}

\noindent Next, we introduce basic mechanisms that preserve differential privacy, which are ingredients in our algorithms. We write $X\sim \mathsf{Lap}(b)$ for any parameter $b>0$ if and only if the density function of $X$ is $p_X(x) = \frac{1}{2b} \exp\left(-\frac{|x|}{b}\right)$.

\begin{lemma}[Laplace mechanism]\label{lem:laplace}
   Fix any $\epsilon>0$. Let $v,v' \in \mathbb{R}^{2k}$ be two arbitrary vectors such that $\|v-v'\|_1\leq s$, then if $X = (X_1\cdots X_{2k})$ are i.i.d. sampled random variables such that $X_i \sim \mathsf{Lap}(s/\epsilon)$, then for any $R\subseteq \mathbb{R}^{2k}$, 
  $$\Pr{v+X\in R} \leq e^\epsilon\Pr{v'+X \in R}.$$
\end{lemma}


In many cases the output domain of the query function is discrete. For example, according to a private dataset, we would like to output a candidate with the highest score. A fundamental mechanism for choosing a candidate is the {\em exponential mechanism}. Let $\mathcal{R}$ be the set of all possible candidates, and let $s:\mathcal{X}\times \mathcal{R}\rightarrow \mathbb{R}_{\geq 0}$ be a scoring function. We define the sensitivity of $s$ be $$\mathsf{sens}_s = \max_{x,y\in \mathcal{X}\atop x\sim y} |s(x) - s(y)|.$$ 
If we wish to output a candidate that minimizes the scoring function, we define the exponential mechanism $\mathcal{E}$ for input dataset $D\in \mathcal{X}$ to be
$\mathcal{M}(D) := \text{Choose a candidate } r\in \mathcal{R}$ { with probability proportional to } $\exp\left(- \frac{{\epsilon \cdot s(D,r)}}{2\mathsf{sens}_s}\right)$ { for all } $r\in \mathcal{R}.$ 
We have the following lemmas for the privacy guarantee and utility of $\mathcal{M}(\cdot)$.
\begin{lemma}\label{lem:exp}
    The exponential mechanism $\mathcal{M}(\cdot)$ sampling in some discrete space $\mathcal{R}$ is $\epsilon$-differentially private.
\end{lemma}

\begin{lemma}\label{l.exp_utility}
    Let $\mathcal{M}(\cdot)$ be an exponential mechanism sampling in some discrete space $\mathcal{R}$. For any input dataset $D$ and $t>0$, let $OPT = \min_{r\in \mathcal{R}}s(D,r)$. Then
    $$Pr\left[s(D,\mathcal{E}(D)) \geq OPT + \frac{2\mathsf{sens}_s}{\epsilon} (\log (|\mathcal{R}|) + t)\right] \leq \exp(-t).$$
\end{lemma}

\section{Missing Proofs}
\subsection{Proof of Lemma \ref{lem:dominating_stcut}.}

\begin{lemma}\label{lem:dominating_stcut_restate} (Restatement of Lemma \ref{lem:dominating_stcut}.)
Let $S \subseteq {n\choose 2}$ such that $S = \{\{u,v\}: u\in \{s,t\}, v\in [n]\backslash \{s,t\}\}$. Then, $S$ is a dominating set of $f$ with sensitivity $2$, where $f$ is the function defined in \Cref{eq:deff_stcut}.
\end{lemma}
\begin{proof}
    (Of Lemma \ref{lem:dominating_stcut}.) We note that $f$ can be written in the following optimization problem:
    $$f(w) = \arg\min_{y\in \{-1,1\}^n} \sum_{\{i,j\}\in {n\choose 2}}w_{\{i,j\}} (1-y_i y_j) = \arg\min_{y\in \{-1,1\}^n} \tilde{f}(w).$$
    with the constraint $y_s = 1$ and $y_t = -1$. Here, we let $\tilde{f}(w) = \sum_{\{i,j\}\in {n\choose 2}}w_{\{i,j\}} (1-y_i y_j)$. We assume $w'$ is a neighboring of $w$ formed by changing the weight of a single edge $\{u,v\}$ in $w$ by $z\in [1,-1]$ ($z$ also needs to make sure that the edge weights of $w'$ are non-negative). Next, we verify that we could always find a correction vector $a$ such that only shifting the edge weights in $S$ is enough to make sure the value of $f$ is unchanged with replacing $x$ to $x'$, and that $a$ depends on only $f(w)$ and $w'-w = (\{u,v\}, z)$.
    
    \noindent \textbf{(1) Suppose $\{u,v\} = \{s,t\}$}. Then changing the weight between $u,v$ does not change $f(w)$ at all. It suffices to just let $a = \mathbf{0}$.
    
    \noindent \textbf{(2) Suppose exactly one of $u,v$ is in $\{s,t\}$}. Since the change happens in the dominating set, it suffices to just let $a_{\{u,v\}} = -z$ and $0$ anywhere else. 

    \noindent \textbf{(3) Suppose both $u,v$ are not in $\{s,t\}$}. In this case, $f(w')$ can be written as
   \begin{align*}
       f(w') &= \arg\min_{y\in \{-1,1\}^n} \left(\sum_{\{i,j\}\in {n\choose 2}}w_{\{i,j\}} (1-y_i y_j) \right) + z(1-y_u y_v)\\
       & = \arg\min_{y\in \{-1,1\}^n} \tilde{f}(w').
   \end{align*}
    Let $y_u^*, y_v^*\in \{1,-1\}$ be the values of $f(w)$ on $u$ and $v$ respectively. Then, it suffices to show that there exists $a_{\{s,u\}}, a_{\{t,v\}}\in [-1,1]$ such that 
    $$\min_{y_u, y_v \in \{1,-1\}}g(z,a_{\{s,u\}}, a_{\{t,v\}}) = \min_{y_u, y_v \in \{1,-1\}}a_{\{s,u\}} (1-y_sy_u) + a_{\{t,v\}} (1-y_ty_v) + z(1-y_u y_v)$$
    achieves it minimum value (not have to be unique) when $y_u = y_u^*$ and $y_v = y_v^*$. This is because that, in this case, the minimal value of 
    \begin{align*}
        \tilde{f}(w') + a_{\{s,u\}} (1-y_sy_u) + a_{\{t,v\}} (1-y_ty_v) &= \tilde{f}(w) + z(1-y_u y_v) + a_{\{s,u\}} (1-y_sy_u) + a_{\{t,v\}} (1-y_ty_v) \\
        & = \tilde{f}(w) + g(z,a_{\{s,u\}}, a_{\{t,v\}})
    \end{align*}
    remains to be attained uniquely at $f(w)$, as $f(w)$ uniquely minimizes $\tilde{f}(w)$ (see also Fact \ref{fac:common_minimum}). We then discuss by different cases:
    \begin{enumerate}
        \item If $z\geq 0$ and $y_u^* = y_v^*$, then let $a_{\{s,u\}} = a_{\{t,v\}} = 0$. In this case, 
        $$g(z,a_{\{s,u\}}, a_{\{t,v\}}) = z(1-y_uy_v) \geq 0$$
        and $g(z,a_{\{s,u\}}, a_{\{t,v\}}) = 0$ when $y_u = y_v = y_u^* = y_v^*$.
        \item If $z \geq 0$ and $y_u^* = 1$, $y_v^* = -1$. Let $a_{\{s,u\}} = a_{\{t,v\}} = 1$. In this case,
        \begin{align*}
            g(z,a_{\{s,u\}}, a_{\{t,v\}}) &= 2 - y_sy_u - y_ty_v + z(1-y_uy_v) \\
            & = 2 + (y_v - y_u) + z(1-y_uy_v) \geq \min \{2, 2z\} = 2z,
        \end{align*}
        and note that when $y_u = 1 $ and $y_v = -1$, $ g(z,a_{\{s,u\}}, a_{\{t,v\}}) = 2z$. The case when $y_u^* = -1$, $y_v^* = 1$ is symmetric.
        \item If $z<0$ and $y_u^* \neq y_v^*$, then let $a_{\{s,u\}} = a_{\{t,v\}} = 0$. In this case, 
        $$g(z,a_{\{s,u\}}, a_{\{t,v\}}) = z(1-y_uy_v) \geq 2z$$
        and $g(z,a_{\{s,u\}}, a_{\{t,v\}}) = 2z$ when $y_u = y_u^*$ and $ y_v = y_v^*$.
        \item If $z < 0$ and $y_u^* =y_v^* = 1$. Let $a_{\{s,u\}} =1$ and $a_{\{t,v\}} = -1$. In this case,
        \begin{align*}
            g(z,a_{\{s,u\}}, a_{\{t,v\}}) &= 1- y_sy_u + y_ty_v - 1 + z(1-y_uy_v) \\
            & = -(y_u + y_v) + z(1-y_uy_v).
        \end{align*}
        If $y_u\neq y_v$, then $g(z,a_{\{s,u\}}, a_{\{t,v\}}) = 2z$. While when $y_u = y_u = 1$, 
        $$g(z,a_{\{s,u\}}, a_{\{t,v\}}) = -2 \leq 2z.$$
        Finally, note that when $y_u = y_u = -1$, $g(z,a_{\{s,u\}}, a_{\{t,v\}}) = 2.$
        Therefore, $g(z,a_{\{s,u\}}, a_{\{t,v\}})$ is minimized on $y_u = y_v = 1$.
         The case when $y_u^* = y_v^* = -1$ is symmetric.
    \end{enumerate}
    In all cases, $g(z,a_{\{s,u\}}, a_{\{t,v\}})$ is minimized on $y_u = y_u^*$ and $ y_v = y_v^*$, and that both the absolute values of $a_{\{s,u\}}$ and $a_{\{t,v\}}$ are bounded by $1$. This completes the proof.
\end{proof}

\subsection{Proof of Lemma \ref{lem:dominating_multi-cut}.}

\begin{lemma}\label{lem:dominating_multi-cut_restate} (Restatement of Lemma \ref{lem:dominating_multi-cut}.)
Let $S \subseteq {V'\choose 2}$ such that $S = \{\{u,v\}: u\in \{s_i,t_i\} (\forall i\in [k]), v\in V'\backslash \{s_i,t_i\}\}$. Then, $S$ is a dominating set of $f$ with sensitivity $2$, where $f$ is the function defined in \Cref{eq:deff_multicut}.
\end{lemma}
\begin{proof}
    (Of Lemma \ref{lem:dominating_multi-cut}.) Similar to the proof of Lemma \ref{lem:dominating_stcut}, we note that without the lose of generality, $f$ can be written as 
   
    \begin{gather*}
    f(w) = \arg\min_{y \in [2k]^{V'}}  \sum_{\{i,j\}\in {V'\choose 2}} w_{ij}\mathbb{I}(y_i \neq y_j) \\
    \begin{aligned}
    \textup{s.t.}\quad & y_{s_1} = c_1, y_{t_1} = c_2; &  \\
                       &y_{s_i} \neq y_{t_i}& \forall 1<i\leq k.
    \end{aligned}
\end{gather*}
Here, $c_1,c_2\in [2k]$ and $c_1\neq c_2$. Again, we let $\tilde{f}(w) = \sum_{\{i,j\}\in {V'\choose 2}} w_{ij}\mathbb{I}(y_i \neq y_j)$. We assume $w'$ is a neighboring of $w$ formed by changing the weight of a single edge $\{u,v\}$ in $w$ by $z\in [1,-1]$. Next, we verify that we could always find a correction vector $a$ that only shift the edge weights in $S$ to make sure the value of $f$ is unchanged with the change, and that $a$ depends on only $f(w)$ and $w'-w = (\{u,v\}, z)$. From the exact same analysis of Lemma \ref{lem:dominating_stcut}, we only consider the case where both $u$ and $v$ are not in $\{s_1,t_1\}$, and that $\{u,v\} \neq \{s_i,t_i\}$ for any $1<i\leq k$ (it does not excludes that $u$ and $v$ are other non-conflicting terminals). This is because that in these cases, either the change does not affect the value of $f$ at all, or the change happens in $S$.

The following proof is almost identical to that of Lemma \ref{lem:dominating_stcut}, but we still present it for the sake of correctness.  Let $y_u^*, y_v^*\in \{1,-1\}$ be the values of $f(w)$ on $u$ and $v$ respectively. Then, it suffices to show that there exists two terminals $s,t\in \{s_1, t_1, \cdots, s_k,t_k\}$ and $a_{\{s,u\}}, a_{\{t,v\}}\in [-1,1]$ such that 
    $$\min_{y_u, y_v \in [2k]}g(z,a_{\{s,u\}}, a_{\{t,v\}}) = \min_{y_u, y_v \in \{1,-1\}}a_{\{s,u\}} \mathbb{I}(y_{s} \neq y_u) + a_{\{t,v\}} \mathbb{I}(y_{t} \neq y_v) + z\mathbb{I}(y_u \neq y_v)$$
    achieves it minimum value when (not have to be unique) $y_u = y_u^*$ and $y_v = y_v^*$, as in this case the unique optimal value of 
    $\tilde{f}(w') +a_{\{s,u\}} \mathbb{I}(y_{s} \neq y_u) + a_{\{t,v\}} \mathbb{I}(y_{t} \neq y_v) $ remains $f(w)$. We then discuss by different cases.
    \begin{enumerate}
        \item If $z\geq 0$ and $y_u^* = y_v^*$, then let $a_{\{s,u\}} = a_{\{t,v\}} = 0$. In this case, 
        $$g(z,a_{\{s,u\}}, a_{\{t,v\}}) = z\mathbb{I}(y_u \neq y_v) \geq 0.$$
        Further, $g(z,a_{\{s,u\}}, a_{\{t,v\}}) = 0$ when $y_u = y_v = y_u^* = y_v^*$.
        \item If $z \geq 0$, $y_u^* = c_1^*$ and $y_v^* = c_2^*$ such that $c_1^*\neq c_2^*$. Then let $s$, $t$ be the terminals that are colored in the optimal solution by $c_1^*$ and $c_2^*$ respectively. Let $a_{\{s,u\}} = a_{\{t,v\}} = 1$. Without the loss of generality, let $c_1^* = 1$ and $c_2^* = -1$. Since $2\mathbb{I}(y_w\neq y_x) = (1-y_wy_x)$ for $y_w,y_x\in \{-1,1\}$, then in this case,
        \begin{align*}
            2\cdot g(z,a_{\{s,u\}}, a_{\{t,v\}}) &= 2 - y_{s_1}y_u - y_{t_1}y_v + z(1-y_uy_v) \\
            & = 2 + (y_v - y_u) + z(1-y_uy_v) \geq \min \{2, 2z\} = 2z,
        \end{align*}
        and note that when $y_u = 1 $ and $y_v = -1$, $ g(z,a_{\{s,u\}}, a_{\{t,v\}}) = 2z$. The case when $y_u^* = -1$, $y_v^* = 1$ is symmetric.
        \item If $z<0$ and $y_u^* \neq y_v^*$, then let $a_{\{s,u\}} = a_{\{t,v\}} = 0$. In this case, 
        $$g(z,a_{\{s,u\}}, a_{\{t,v\}}) = z\mathbb{I}(y_u\neq y_v) \geq z$$
        and $g(z,a_{\{s,u\}}, a_{\{t,v\}}) = z$ when $y_u = y_u^*$ and $ y_v = y_v^*$, since $\mathbb{I}(y_u^*\neq y_v^*) = 1$.
        \item If $z < 0$ and $y_u^* =y_v^* = c^*$. Let $s_i$ be the terminal that is colored in the optimal solution by $c^*$. Without the loss of generality, let $c^* = 1$, and let $t_i$ be colored by $-1$. Let $a_{\{s_i,u\}} =1$ and $a_{\{t_i,v\}} = -1$. In this case,
        \begin{align*}
            2\cdot g(z,a_{\{s_i,u\}}, a_{\{t_i,v\}}) &= 1- y_{s_i}y_u + y_{t_i}y_v - 1 + z(1-y_uy_v) \\
            & = -(y_u + y_v) + z(1-y_uy_v).
        \end{align*}
        If $y_u\neq y_v$, then $g(z,a_{\{s_i,u\}}, a_{\{t_i,v\}}) = 2z$. While when $y_u = y_u = 1$, 
        $$g(z,a_{\{s_i,u\}}, a_{\{t_i,v\}}) = -2 \leq 2z.$$
        Finally, note that when $y_u = y_u = -1$, $g(z,a_{\{s_i,u\}}, a_{\{t_i,v\}}) = 2.$
        Therefore, $g(z,a_{\{s_i,u\}}, a_{\{t_i,v\}})$ is minimized on $y_u = y_v = 1 = c^*$.
         The case when $y_u^* = y_v^* = -1$ is symmetric.
    \end{enumerate}
    In all cases, $g(z,a_{\{s_1,u\}}, a_{\{t_1,v\}})$ is minimized on $y_u = y_u^*$ and $ y_v = y_v^*$, and that both the absolute values of $a_{\{s_1,u\}}$ and $a_{\{t_1,v\}}$ are bounded by $1$. This completes the proof.
\end{proof}

\subsection{Proof of Lemma \ref{lem:dominating_maxstcut}.}
\begin{lemma}\label{lem:dominating_maxstcut_restate} (Restatement of Lemma \ref{lem:dominating_maxstcut}.) Let $\hat{V} = [n]\cup \{s,t\}$, and let $S \subseteq {\hat{V}\choose 2}$ such that $S = \{\{u,v\}: u\in \{s,t\}, v\in [n]\}$. Then, $S$ is a dominating set of $\hat f$ with sensitivity $2$. Here, $\hat{f}$ is the function defined in \Cref{eq:deff_maxstcut}.
\end{lemma}
\begin{proof}
    (Of Lemma \ref{lem:dominating_maxstcut}.) Similar to the proof of Lemma \ref{lem:dominating_stcut}, we note that $\hat f$ can be written in the following optimization problem:
    $$\hat f(w) = \arg\max_{y\in \{-1,1\}^{\hat{V}}} \sum_{\{i,j\}\in {\hat{V}\choose 2}}w_{\{i,j\}} (1-y_i y_j).$$
    with the constraint $y_s = 1$ and $y_t = -1$. Again, we assume $w'$ is a neighboring of $w$ formed by changing the weight of a single edge $\{u,v\}$ in $w$ by $z\in [1,-1]$. Still, we verify that we could always find a correction vector $a$ that only shift the edge weights in $S$ to make sure the value of $f$ is unchanged with the change, and that $a$ depends on only $f(w)$ and $w'-w = (\{u,v\}, z)$. Without the loss of generality, we assume both $u$ and $v$ are not in $\{s,t\}$. Otherwise the construction of the correction vector $a$ is trivial. Let $y_u^*, y_v^*\in \{1,-1\}$ be the values of $\hat{f}(w)$ on $u$ and $v$ respectively. For the same reason as in the proof of Lemma \ref{lem:dominating_stcut}, it suffices to show that there exists $a_{\{s,u\}}, a_{\{t,v\}}\in [-1,1]$ such that 
    $$\max_{y_u, y_v \in \{1,-1\}}g(z,a_{\{s,u\}}, a_{\{t,v\}}) = \max_{y_u, y_v \in \{1,-1\}}a_{\{s,u\}} (1-y_sy_u) + a_{\{t,v\}} (1-y_ty_v) + z(1-y_u y_v)$$
    achieves it maximum value when (not have to be unique) $y_u = y_u^*$ and $y_v = y_v^*$. The following proof is symmetric to the proof of Lemma \ref{lem:dominating_stcut}, we include it for the sake of correctness. Next, we discuss by different cases:
    \begin{enumerate}
        \item If $z\geq 0$ and $y_u^* \neq y_v^*$, then let $a_{\{s,u\}} = a_{\{t,v\}} = 0$. In this case, 
        $$g(z,a_{\{s,u\}}, a_{\{t,v\}}) = z(1-y_uy_v) \leq 2z$$
        and $g(z,a_{\{s,u\}}, a_{\{t,v\}}) = 2z$ when $y_u =y_u^*$ and $y_v = y_v^*$.
        \item If $z \geq 0$ and $y_u^* =y_v^* = 1$. Let $a_{\{s,u\}} =-1$ and $a_{\{t,v\}} = 1$. In this case,
        \begin{align*}
            g(z,a_{\{s,u\}}, a_{\{t,v\}}) &= -1 +y_sy_u +1- y_ty_v + z(1-y_uy_v) \\
            & = y_u +y_v + z(1-y_uy_v).
        \end{align*}
        Note that when both $y_u = y_v = 1$, $g(z,a_{\{s,u\}}, a_{\{t,v\}}) = 2$. 
        If $y_u\neq y_v$, then $$g(z,a_{\{s,u\}}, a_{\{t,v\}}) = 2z \leq 2.$$ And when both $y_u = y_v = -1$, $g(z,a_{\{s,u\}}, a_{\{t,v\}}) = -2$. Therefore, $g(z,a_{\{s,u\}}, a_{\{t,v\}})$ is minimized on $y_u = y_v = 1$. The case when $y_u^* = y_v^* = -1$ is symmetric.
        \item If $z<0$ and $y_u^* = y_v^*$, then let $a_{\{s,u\}} = a_{\{t,v\}} = 0$. In this case, 
        $$g(z,a_{\{s,u\}}, a_{\{t,v\}}) = z(1-y_uy_v) \leq 0$$
        and $g(z,a_{\{s,u\}}, a_{\{t,v\}}) = 0$ when $y_u = y_u^* = y_v = y_v^*$.
        \item If $z < 0$ and $y_u^* = 1$, $y_v^* = -1$. Let $a_{\{s,u\}} =-1$ and $a_{\{t,v\}} = -1$. In this case,
        \begin{align*}
            g(z,a_{\{s,u\}}, a_{\{t,v\}}) &= -1+ y_sy_u -1 +y_ty_v  + z(1-y_uy_v) \\
            & = -2 + (y_u - y_v) + z(1-y_uy_v) \leq \max\{-2,2z\} = 2z
        \end{align*}
        since $z \geq -1$. And note that when $y_u = 1 $ and $y_v = -1$, $ g(z,a_{\{s_1,u\}}, a_{\{t_1,v\}}) = 2z$.
         The case when $y_u^* = -1$, $y_v^* = 1$ is symmetric.
    \end{enumerate}
    In all cases, $g(z,a_{\{s,u\}}, a_{\{t,v\}})$ is maximized on $y_u = y_u^*$ and $ y_v = y_v^*$, and that both the absolute values of $a_{\{s,u\}}$ and $a_{\{t,v\}}$ are bounded by $1$. This completes the proof.
\end{proof}

\end{document}